\documentclass[11pt, letterpaper]{article}
\usepackage[utf8]{inputenc}
\usepackage{fullpage,times}

\usepackage{amsmath}
\usepackage{amsthm}
\usepackage{amssymb}
\usepackage{amsfonts}
\usepackage{graphicx}
\usepackage{setspace}
\usepackage{hyperref}
\usepackage{caption}
\usepackage{url}
\usepackage[dvipsnames]{xcolor}
\usepackage{booktabs}
	
\usepackage{color}
\definecolor{darkblue}{rgb}{0.0,0.0,0.65}
\definecolor{darkred}{rgb}{0.65,0.0,0.0}
\hypersetup{
  colorlinks = true,
  citecolor  = darkblue,
  linkcolor  = darkred,
  filecolor  = darkblue,
  urlcolor   = darkblue,
}

\usepackage[ruled,vlined,linesnumbered]{algorithm2e}
\SetKw{Continue}{continue}
\SetKw{Or}{or}
\SetKw{Pass}{PASS}
\SetKw{Fail}{FAIL}
\SetKwFunction{MinDist}{min-distance}
\SetKwFunction{MinDiam}{min-diameter}
\SetKwFunction{BichromaticMinDiam}{bichromatic-min-diameter}
\SetKwData{Left}{left}\SetKwData{This}{this}\SetKwData{Up}{up}
	
\usepackage[capitalise]{cleveref}
	
\def\dm{d_{\text{min}}}
\def\eps{\varepsilon}

\newtheorem{theorem}{Theorem}

\newtheorem{lemma}[theorem]{Lemma}
\newtheorem*{remark}{Remark}

\newtheorem*{ovc}{OV Hypothesis}
\newtheorem{proposition}[theorem]{Proposition}
\newcommand{\andt}{\textup{\textbf{ and }}}

\newenvironment{reminder}[1]{\smallskip
\noindent {\bf Reminder of #1 }\em}{\smallskip}

\title{Approximating Min-Diameter: Standard and Bichromatic} 

\date{August 16, 2023}

\author{Aaron Berger\thanks{\url{bergera@mit.edu}. Massachusetts Institute of Technology.} \qquad Jenny  Kaufmann\thanks{\url{jkaufmann@math.harvard.edu}. Harvard University. Supported by the National Science Foundation Graduate Research Fellowship under Grant No. DGE 2140743.} \qquad Virginia Vassilevska Williams\thanks{\url{virgi@mit.edu}. Massachusetts Institute of Technology. Supported by NSF Grant CCF-2129139 and a Sloan Research Fellowship.}}

\begin{document}

\maketitle

\begin{abstract}
The \textit{min-diameter} of a directed graph $G$ is a measure of the largest distance between nodes. It is equal to the maximum min-distance $\dm(u,v)$ across all pairs $u,v \in V(G)$, where $\dm(u,v) = \min(d(u,v), d(v,u))$.  Min-diameter approximation in directed graphs has attracted attention recently as an 
offshoot of the classical and well-studied diameter approximation problem.

Our work provides a $O(m^{1.426}n^{0.288})$-time $\frac{3}{2}$-approximation algorithm for min-diameter in DAGs, and a faster $O(m^{0.713}n)$-time almost-$\frac{3}{2}$-approximation variant. (An almost-$\alpha$-approximation algorithm determines the min-diameter to within a  multiplicative factor of $\alpha$ plus constant additive error.) This is the first known algorithm to solve $\frac{3}{2}$-approximation for min-diameter in sparse DAGs in \textit{truly subquadratic} time $O(m^{2-\epsilon})$ for $\epsilon > 0$; previously only a $2$-approximation was known. By a conditional lower bound result of [Abboud et al, SODA 2016], a better than $\frac{3}{2}$-approximation can't be achieved in truly subquadratic time under the Strong Exponential Time Hypothesis (SETH), so our result is conditionally tight. 
We additionally obtain a new conditional lower bound for min-diameter approximation in general directed graphs, showing that under SETH, one cannot achieve an approximation factor below 2 in truly subquadratic time.

Our work also presents the first study of approximating bichromatic min-diameter, which is the maximum min-distance between oppositely colored vertices in a 2-colored graph. We show that SETH implies that in DAGs, a better than 2 approximation cannot be achieved in truly subquadratic time, and that in general graphs, an approximation within a factor below $\frac{5}{2}$ is similarly out of reach. We then obtain an $O(m)$-time algorithm which determines if bichromatic min-diameter is finite, and an almost-2-approximation algorithm for bichromatic min-diameter with runtime $\tilde{O}(\min(m^{4/3}n^{1/3}, m^{1/2}n^{3/2}))$.

\end{abstract}

\section{Introduction}

The \textit{min-distance} between two vertices $x, y$ in a directed graph $G$ is the minimum of the one-way distances $d(x,y)$ and $d(y,x)$, and is written $\dm(x,y)$. This notion of distance was introduced by Abboud, Vassilevska W., and Wang \cite{avw} in their study of diameter in directed graphs. Since the standard notion of distance in directed graphs is not symmetric, \cite{avw} considered a number of symmetric distance functions: roundtrip distance, max-distance, and min-distance. The min-distance in particular has since then been studied in a series of papers \cite{minajenny,d19,cz22}.

The min-distance is arguably the most natural notion of distance in directed acyclic graphs (DAGs), in which for every two vertices $x,y$, at most one of $d(x,y)$, $d(y,x)$ is finite. Min-distance is also applicable in potential real-world contexts: for example, if a patient needs to see a doctor as soon as possible, the doctor can visit the patient or vice versa.

The \textit{min-diameter} of a directed graph $G$ is the maximum min-distance between any two vertices, or $\max_{x,y \in V(G)} \dm(x,y)$. One can additionally define the \textit{bichromatic min-diameter} in a graph $G$ with vertex set $V = A \sqcup B$ partitioned into ``red'' and ``blue'' vertices: the bichromatic min-diameter is the maximum min-distance between oppositely-colored vertices, or $\max_{a\in A, b \in B} \dm(a,b)$. These are variants on the standard notions of diameter \cite{AbboudGW15, aingworth, BKM95, ChechikLRSTW14,  chung, FHW12,   Yuster10} and bichromatic diameter \cite{BackursRSWW18, bichromatic}, respectively. 

One can compute min-diameter or bichromatic min-diameter -- and for that matter All Pairs Shortest Paths (APSP), the shortest path distances $d(x,y)$ for all $x, y \in V$ -- in an $n$-vertex $m$-edge graph in $O(mn+n^2\log n)$ time, simply by running Dijkstra's algorithm from every vertex. In unweighted graphs, one can instead use BFS, giving an $O(mn)$ runtime. 

One might ask whether, for either min-diameter or bichromatic min-diameter, a faster algorithm exists. However, just as for standard diameter and other diameter variants that have been studied \cite{avw, bichromatic, RV13}, the Strong Exponential Time Hypothesis (SETH) \cite{ipz1,seth2} suggests exact computation cannot be done in runtimes that are \textit{truly subquadratic}, meaning $O(m^{2-\epsilon})$ for $\epsilon > 0$. SETH is one of the main hypotheses in Fine-Grained Complexity \cite{vsurvey}, and is among the most well-established hardness hypotheses for showing conditional lower bounds. It states that for every $\epsilon>0$, there is an integer $k\geq 3$ so that $k$-SAT on $n$ variables cannot be solved in $O(2^{(1-\epsilon)n})$ time.

Since exact computation is conditionally hard, we resort to finding approximations. Since min-distance does not obey the triangle inequality, approximating min-diameter is especially challenging, in comparison to standard diameter or even roundtrip diameter or max-diameter. For this reason, much of the work on min-diameter has focused on DAGs.

\subsubsection*{Min-Diameter}

Abboud, Vassilevska W., and Wang \cite{avw} gave a $2$-approximation for min-diameter in DAGs, running in
$\tilde{O}(m)$ time.\footnote{The tilde hides polylogarithmic factors.} Meanwhile, they showed that if one can obtain a $(\frac{3}{2}-\delta)$ approximation in $O(m^{2-\epsilon})$ time  for min-diameter in DAGs (for $\eps,\delta>0$), then SETH is false. 
The results of \cite{avw} left a gap between the conditional lower bound of $3/2$ and the upper bound of $2$ for $O(m^{2-\epsilon})$ time algorithms for DAGs. (The lower bound is for $O((mn)^{1-\epsilon})$ time algorithms but since the instances are sparse, this is the same as $O(m^{2-\epsilon})$ time.)

Later work by Dalirrooyfard and Kaufmann \cite{minajenny} showed that in \textit{dense} DAGs, one can beat the $mn$ barrier by obtaining an almost-$\frac{3}{2}$-approximation algorithm running in $O(n^{2.35})$ time. 
The conditional lower bound of \cite{avw} is for sparse DAGs, however, and the gap between upper and lower bounds has remained.

Our main result is to close this gap for sparse DAGs:

\begin{theorem}
There is an $O(m^{0.713}n)$ time algorithm that achieves a $(\frac{3}{2},\frac{1}{2})$-approximation for min-diameter in any $m$-edge, $n$-node unweighted DAG.

Furthermore, there is an $O(m^{1.426}n^{0.288})$ time algorithm that achieves a $\frac{3}{2}$-approximation for min-diameter in any $m$-edge, $n$-node unweighted DAG.
\end{theorem}

A $(c,a)$-approximation for a quantity $D$ is a quantity $D'$ such that $D\leq D'\leq cD+a$.

Similar to \cite{minajenny}, we use hitting set and set intersection methods to certify distances. Our main new technique is to iteratively grow a central interval of vertices with convenient distance properties by checking that at least one of two neighboring vertex subsets, to its left and right, has the desired properties.

The algorithms use fast matrix multiplication. There are also combinatorial versions of the algorithms, with runtimes $O(m^{3/4}n)$ for the $(\frac{3}{2}, \frac{1}{2})$-approximation and $O(m^{5/4}n^{1/2})$ for the $\frac{3}{2}$-approximation.\footnote{Combinatorial is here used to refer to an algorithm that does not rely on fast matrix multiplication.} These runtimess are still subquadratic for sparse graphs.

Our paper also considers the case of general directed graphs.
Abboud, Vassilevska W., and Wang \cite{avw} showed that under SETH, any $O(m^{2-\eps})$ time algorithm for $\eps>0$ for min-diameter can achieve at best a $2$-approximation.
This result only held for weighted graphs.

We give the first hardness result for unweighted graphs, extending the hardness of \cite{avw}:

\begin{theorem}
Under SETH, there can be no $O(m^{2-\eps})$ time $(2-\delta)$-approximation algorithm for the min-diameter of an unweighted directed graph with $n$ vertices and $m=n^{1+o(1)}$ edges, for $\epsilon, \delta > 0$.
\end{theorem}

 Because our  min-diameter approximation algorithms for DAGs obtain an approximation factor better than 2 in truly subquadratic time, this gives the first separation of hardness results 
for min-diameter approximation in the sparse cyclic versus acyclic cases.

This result, along with all other hardness results from SETH in this work, are via Orthogonal Vectors (OV) reductions. The OV problem is as follows: Given two sets $A, B$ each containing $n$ $d$-dimensional Boolean vectors, determine whether there are vectors $a\in A$ and $b\in B$ that are \textit{orthogonal}, meaning $a \cdot b = 0$.

\begin{ovc}[\cite{w05}]
There is no constant $\epsilon > 0$ such that for any constant $c$ and $d=c\log{n}$, the OV problem can be solved by a randomized algorithm in time $O(n^{2-\epsilon})$.
\end{ovc} 

The OV Hypothesis is implied by SETH \cite{w05}. OV-based sparse graph constructions have been commonly used to provide hardness results for diameter variants for the past decade, after having been first introduced by Roditty and Vassilevska W. in \cite{RV13}. However, unlike previous OV-based constructions known to us, our construction is of a graph shaped as a cycle. If there are no orthogonal vectors then all vertices can reach one another via a single loop around the cycle, whereas if there are orthogonal vectors a path between them must loop nearly twice around the cycle, giving a min-diameter nearly twice as large.

As for upper bounds for general directed graphs, Dalirrooyfard et al. \cite{d19} gave for every integer $k\geq 2$, an $\tilde{O}(mn^{1/k})$ time $(4k-5)$-approximation algorithm for min-diameter. Most recently, Chechik and Zhang \cite{cz22} achieved a $4$-approximation in near-linear time. There is still  a gap between the best approximation known in $O(m^{2-\eps})$ time ($3$ by \cite{d19}) and the best hardness result for such algorithms ($2$ by this work for unweighted, and by \cite{avw} for weighted).

\subsubsection*{Bichromatic Min-Diameter}

A \textit{bichromatic} version of diameter was considered by Dalirrooyfard et al. \cite{bichromatic}. In a graph whose nodes are colored red and blue, the \textit{bichromatic diameter} is the largest distance between two nodes of different colors. Dalirrooyfard et al. \cite{bichromatic} gave algorithms and hardness for bichromatic diameter under the usual notion of distance. 

Our work presents the first study of the notion of bichromatic min-diameter. As the min-diameter is the most natural diameter notion for DAGs, the bichromatic min-diameter is likewise a natural way to consider distances in DAGs whose vertices are two-colored.

We first give hardness for general directed graphs, suggesting that bichromatic min-diameter may be harder than regular min-diameter:

\begin{theorem}
Under SETH, there can be no $O(m^{2-\eps})$ time $(\frac 52-\delta)$-approximation algorithm for bichromatic min-diameter in unweighted  $n$-node, $m=n^{1+o(1)}$-edge graphs for $\epsilon, \delta > 0$.

Furthermore, under SETH,
there can be no $O(m^{2-\eps})$ time $(3-\delta)$-approximation algorithm for bichromatic min-diameter in graphs with $O(\log n)$ bit integer edge weights. 
\end{theorem}

It would be interesting if the $\tilde{O}(m\sqrt n)$ $3$-approximation algorithm for min-diameter of \cite{d19} can be extended to work for bichromatic min-diameter, as then one could get a tight result in the weighted case.

We next turn to bichromatic min-diameter in DAGs. We give an almost-$2$-approximation algorithm and show that it is essentially tight (up to the additive error) under SETH:

\begin{theorem}
Under SETH, there can be no $O(m^{2-\eps})$ time $(2-\delta)$-approximation algorithm for bichromatic min-diameter in unweighted $n$-node, $m=n^{1+o(1)}$-edge DAGs, for $\epsilon, \delta > 0$.
\end{theorem}

\begin{theorem}\label{thm:main-dag}
There is an $\tilde{O}(m^{4/3}n^{1/3})$-time algorithm, which, given a DAG and maximum red-blue edge weight $M$, outputs a $(2, M)$-approximation of the bichromatic min-diameter.
\end{theorem}

Finally, we present a linear-time algorithm which determines whether a directed graph has finite bichromatic min-diameter. 

\begin{theorem}\label{thm:bichr-finite}
There is an $O(m)$ time algorithm which checks, for any weighted directed graph $G$, whether the bichromatic min-diameter is finite.
\end{theorem}

\subsection{Preliminaries}
We assume the word-RAM model of computation with $O(\log n)$ bit words. All of our algorithms and reductions fall within this model.

Graphs are directed and weakly connected.\footnote{If a graph is not weakly connected (which can  be checked in $O(m+n)$ time), then it has infinite min-diameter, as well as infinite bichromatic min-diameter if its vertices are 2-colored.} Edge weights are polynomial in $n$.

The min-eccentricity of a vertex, $\epsilon(v)$, is given by $\max_{u \in V} \dm(u,v)$.

For  $v \in V$, the \textit{distance-$D$ out-neighborhood} of $v$ is $N^{out}_D(v) = \{w \in V \setminus \{v\} \ | \ d(v, w) \leq D\}$. We define $N^{in}_D(v)$ correspondingly. 

Given a DAG $G$ with topological ordering $\pi$ and vertex sets $S, T \subseteq V = V(G)$, we write $S <_\pi T$ if all vertices in $S$ appear to the left of all vertices in $T$. When $S$ or $T$ is equal to $\{x\}$ for some vertex $x$, we may omit the brackets. For vertices $s, t$, we write $s \leq_\pi t$ if $s<_\pi t$ or $s=t$. We define a closed subset (with respect to $\pi$) to be a subset $S$ such that for all $v \in V$, either $v \in S$, $v >_\pi S$, or $v <_\pi S$.

Given a DAG $G$ with topological ordering $\pi$, a vertex $v \in V$ and a set $S \subseteq V$, let $s_v \in S$ be the left-most vertex in $S$ such that $d(v, s) \leq D$, if such an $s_v$ exists. Then we define $N^{out}_{D, S}(v)$ to be the set of vertices $w$ such that $d(v, w) \leq D$ and, if $s_v$ exists, $w  \leq_\pi s_v$.  One can intuitively think of $N^{out}_{D, S}(v)$ as the set $N^{out}_D(v)$ of vertices  at distance at most $D$ from $v$, but cut off after the first (left-most) time we hit $S$. We define $N^{in}_{D, S}(v)$ symmetrically. A set $S$ such that for all $v$, $|N^{out}_{d,S}(v)|, |N^{in}_{d',S}(v)| \leq k$, will be called a $(k, (d, d'))$-\textit{neighborhood cover}. If $d = d'$ we refer to it as a $(k,d)$-neighborhood cover.

A \textit{bichromatic DAG} $G$ is a DAG whose vertices are two-colored. An \textit{$(A,B)$-separated DAG}, which we may also simply call a \textit{separated DAG} is a DAG ordered according to some topological ordering $\pi$ with color sets $A, B$ such that $A <_\pi B$.

We sometimes omit $\pi$ when the choice of $\pi$ is clear.

Let $\omega(1, r, 1)$ be the exponent of the runtime of multiplying $n \times n^{r}$ by $n^{r} \times n$ matrices. The square matrix multiplication exponent is $\omega = \omega(1, 1, 1) > 2.3715 $ \cite{williams2023new,duan2023}.

\subsection{Techniques}

In this section, we review two useful techniques. The first of these is the greedy set cover lemma. This lemma, and a related randomized version, have been commonly used in prior work on diameter variants (\cite{aingworth}, \cite{RV13}, \cite{CGR}, \cite{avw}, \cite{minajenny}). See \cite{hittingset} for a proof of the lemma.

\begin{lemma}\label{setcover}
Let $p = O(n)$, and let $X_1, \dots X_p \subseteq V$ have size $|X_i| \geq n^\epsilon$ for $\epsilon \in [0, 1]$. In time $O(n^{1+\epsilon})$ one can construct a set $S \subseteq V$ of size $O(n^{1-\epsilon}\log n)$ such that $S \cap X_i \neq \varnothing$ for all $i \in [p]$.
\end{lemma}

The following technique, previously used in \cite{minajenny}, constructs a set cover of all sufficiently large balls of radius $d$.

\begin{lemma}\label{nbrhds}
Given a topologically ordered DAG $G$ and parameters $d, d'$ and $k = n^{\epsilon}$ for $\epsilon \in [0, 1]$, one can in time $O(nk^2)$ construct a $(k, (d, d'))$-neighborhood cover $S$ of size $O(\frac{n}{k}\log n)$, and also construct the sets  $N^{out}_{d,S}(v)$, $N^{in}_{d',S}(v)$ for all $v \in V$.
\end{lemma}

\begin{proof}
 For each vertex $v$, if $|N^{out}_{d}(v)| < k$, we define $X^k_{d}(v) = N^{out}_{d}(v)$. Otherwise let $X^k_{d}(v)$ be the $k$ left-most vertices in $N^{out}_{d}(v)$. Similarly, we define the set $Y^k_{d'}(v)$: if $|N^{in}_{d'}(v)| < k$, then $Y^k_{d'}(v) = N^{in}_{d'}(v)$, and otherwise $Y^k_{d'}(v)$ is the $k$ right-most vertices in $N^{in}_{d'}(v)$.
 
 We can compute $X^k_{d}(v)$ as follows: We initialize a set $S_0 = \varnothing$. 
 While $|S_t| < k$, at step $t$, we consider the set $W = N^{out}(S_t \cup \{v\}) \cap N_d^{out}(v)$ of out-neighbors of $S_t \cup \{v\}$ that are at distance at most $d$ from $v$. If $W$ is nonempty, we let $w$ be its left-most element, and we construct $S_{t+1} = S \cup \{w\}$. If $W$ is empty, then $S_t = N^{out}_d(v) = X^k_{d}(v)$. If $|S_t| = k$, then $S_t$ contains the left-most $k$ vertices in $N^{out}_d(v)$. Thus, in either case, we will eventually construct $S_t = X^k_{d}(v)$. The key step in this construction process, namely finding $w$, can be done by maintaining a list containing the left-most neighbor of each vertex in $S_t$ such that the neighbor is of distance at most $d$ from $v$, and choosing the left-most vertex in the list at each step. At step $t$, this list has length at most  $|S_t| < k$, and the total number of steps is at most $k$, so the construction can be done in $O(k^2)$ time.
 
 We can construct $Y^k_{d'}(v)$ in $O(k^2)$ time similarly. Doing this for all $v \in V$ takes time $O(nk^2)$.

Lemma \ref{setcover} gives us a set cover $S$ of size $O(\frac{n}{k}\log n)$ which intersects all sets $X^k_{d}(v), Y^k_{d'}(v)$ of size at least $k$. This takes time $O(nk)$. Then for each vertex $v$, we construct $N^{out}_{d,S}(v)$ as the set obtained from $X^k_{d}(v)$ by removing all vertices to the right of the left-most $s \in S \cap X^k_{d}(v)$. We construct $N^{in}_{d,S}(v)$ in a symmetric fashion. Since the sets $X^k_{d}(v), Y^{k}_{d'}(v)$ were of size at most $k$, the sets $N^{out}_{d,S}(v), N^{in}_{d',S}(v)$ are also of size at most $k$.
\end{proof}

\section{Min-diameter approximation}

In Section \ref{section:cond-lower-bound} we present a conditional lower bound showing that the OV Hypothesis implies that no $(2-\delta)$-approximation algorithm for min-diameter in unweighted graphs can run in truly subquadratic time. Subsequently, in Section \ref{section:near32appx}, we give an almost-$\frac{3}{2}$-approximation algorithm for min-diameter in unweighted DAGs which runs in truly subquadratic time.

\subsection{Conditional lower bound in general graphs}\label{section:cond-lower-bound}
We first present a conditional lower bound for approximating min-diameter in general graphs.

\begin{theorem}\label{thm:2-approx-lb}
If there are $\epsilon, \delta > 0$ such that there is an $O(m^{2 - \epsilon})$-time $(2-\delta)$-approximation algorithm for min-diameter, then the OV Hypothesis is false.
\end{theorem}

\begin{proof}
Given $\delta > 0$, choose $t = \lceil \frac{2}{\delta} \rceil$ so that $2 - \delta < \frac{2t}{t+1}$. Let $A$ be an instance of single-set OV\footnote{The OV Hypothesis can be equivalently stated in terms of single-set OV (OV where $A = B$). Informally, the reduction is to construct $A'$ from $A$ by appending $10$ to all vectors, and $B'$ from $B$ by appending $01$ to all vectors. If $v_1, v_2 \in A' \cup B'$ are orthogonal, then we must have $v_1 \in A'$, $v_2 \in B'$ or vice versa.}, that is, a set of $n \geq 2^t$ vectors in $\{0, 1\}^{c_0\log n}$ for a constant $c_0 > 0$. We will construct a graph $G_t$ with $O(tn)$ vertices, $\tilde{O}(tn)$ edges such that the min-diameter is $2t+1$ if $A$ contains a pair of orthogonal vectors, and $t+1$ otherwise. 
For each $i \in [t]$, we construct a set $A_i$ of $n$ vertices corresponding bijectively to the vectors in $A$. For each vector $a \in A$, let $a_i$ be the vertex in set $A_i$ corresponding to $a$. We also construct a set $I$ of ``index'' vertices $x_1, \dots, x_{c\log n}$. In total we have $O(tn)$ vertices. We add an edge $a_1 \to x_j$ and an edge $x_j \to a_t$ whenever $a[j] = 1$. We also add edges $a_i \to a_{i-1}$ for each $i \in \{2, \dots, t\}$. Finally, we add all possible edges $A_2 \to I$. In total we have $\tilde{O}(tn)$ edges. This graph $G_t$ may be constructed in $\tilde{O}(tn) = \tilde{O}(n)$ time.

We now check that if the OV instance is a YES instance, $G_t$ has at least min-diameter $2t+1$, and in the NO case, $G_t$ has min-diameter $t+1$. 

\begin{figure}[h]
    \centering
\includegraphics[scale=0.4]{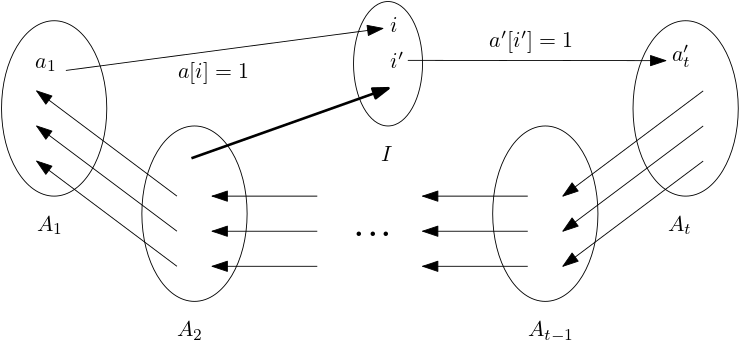}
    \caption{The graph $G_t$. The thick edge denotes that all possible edges $A_2 \to I$ exist.}
\label{fig:min-diam-lb}
\end{figure}

\paragraph*{YES case}
Let $a, b \in A$ be orthogonal. Without loss of generality, $\dm(a_1,b_1) = d(a_1,b_1)$. 
There is no $j$ such that $a[j] = b[j] = 1$, so there is no length-2 path from $a_1$ to $b_t$ via $I$. Since all edges between $A_i$ and $A_{i-1}$ are of the form $a_i \to a_{i-1}$, the first $t$ vertices on any $a_1 \to b_1$ path must be of the form $a_1 \to x_j \to c_t \to \dots \to c_2$ for some $j \in I, c \in A$ such that $a[j] = c[j] = 1$. We note that any path from $c_2$ to $b_1$ must traverse all sets of the cycle with the possible exception of $A_1$, so must be of length at least $t$. Then if $k$ is an index such that $b[k] = 1$, the path $c_2 \to x_k \to b_t \to \dots \to b_1$ gives $d(c_2, b_1) = t+1$. Thus, $d(a_1, b_1) = \dm(a_1, b_1) \geq 2t+1$. So $G_t$ has min-diameter at least $2t+1$.

\paragraph*{NO case}
Suppose that $A$ contains no pair of orthogonal vectors. For any $a_1, b_i$, there is some $j$ such that $a[j] = b[j] = 1$, so there is a path $a_1 \to x_j \to b_t \to \dots b_i$ of length at most $t+1$. Now, consider any vertices $a_i, b_{i'}$ with $2 \leq i \leq i'$. Let $j$ be an index such that $b[j] = 1$. There is a path $a_i \to \dots a_2 \to x_j \to  b_t \to \dots \to b_{i'}$ which has length at most $t$. Furthermore, for any $a_1 \in A_1$, $x_j \in I$, $d(a_1, x_j) \leq t+1$, since there is a path $a_1 \to x_k \to a_t \to \dots \to a_2 \to x_j$ for some $k$ with $a[k]=1$. Lastly, for any $a_i \in A_i$ with $i \geq 2$, and for any  $x_j \in I$, there is a path $a_i \to \dots a_2 \to x_j$ of length at most $t$. Hence, all pairs of vertices are at min-distance at most $t+1$.
\end{proof}

\subsection{A ($\frac{3}{2}$, $\frac{1}{2}$)-approximation algorithm in unweighted DAGs}\label{section:near32appx}
Unlike in general directed graphs, where $(2-\delta)$-approximating min-diameter seems to be hard, in DAGs one can achieve an efficient $(\frac{3}{2}, \frac{1}{2})$-approximation and a (slightly less) efficient $\frac{3}{2}$-approximation. The $(\frac{3}{2}, \frac{1}{2})$ and the exact $\frac{3}{2}$ approximations have very similar proofs; we present the proof of the $(\frac{3}{2}, \frac{1}{2})$-approximation result here because the algorithm is faster and marginally simpler than in the exact $\frac{3}{2}$ case, whose proof can be found in Appendix \ref{app:3/2-mindiam}.

Our algorithm uses fast sparse matrix multiplication to compute set intersections. Its runtime involves the constants $\alpha = \max\{0 \leq r \leq 1 \ | \ \omega(1,r,1) = 2\}$ and  $\beta = \frac{\omega - 2}{1 - \alpha}$. 

\begin{theorem}\label{thm:3/2-unweighted-dag}
There is an $\tilde{O}(m^{\frac{4\beta + 2 - 2\alpha\beta}{5\beta+3-2\alpha\beta}}n)$-time algorithm achieving a $(\frac{3}{2} , \frac{1}{2})$-approximation for min-diameter in unweighted DAGs.
\end{theorem}


Since $\alpha > 0.32133$ and $\omega < 2.37156$ \cite{williams2023new}, we can use $\beta \simeq 0.5475$, giving the runtime $O(m^{0.713}n)$.

The algorithm starts by topologically sorting the DAG and fixing a neighborhood-size parameter $k$ (whose value will be determined later). It then performs two layers of recursion. The outer layer is a binary search over min-diameter estimates $D$: for each estimate $D$, it will either determine that $D$ is low or high, i.e. that $\MinDiam{G} > D$ or $\MinDiam{G} \leq \lceil \frac{3}{2}D\rceil$. The inner layer, in which the estimate $D$ is fixed, involves recursively splitting the graph in half according to the topological ordering. Then, in the core body of the algorithm, one of the following will occur:

\begin{itemize}
    \item We find a pair of vertices with min-distance larger than $D$, in which case we end the recursion and report that $\MinDiam{G} > D$.
    \item We verify that every min-distance between a vertex in the left half and a vertex in the right half is at most $\lceil 3D/2 \rceil$, and we recurse on the left and right halves of the graph.
\end{itemize}

The core body of the algorithm -- which will be repeated recursively as the graph is repeatedly cut in half -- is as follows:

First, using Lemma \ref{nbrhds}, we compute a $(k, (D/2, \lceil D/2 \rceil ))$-neighborhood cover $S$ of size $O(\frac{n}{k}\log n)$ along with neighborhoods $N^{out}_{D/2, S}(v), N^{in}_{\lceil D/2 \rceil, S}(v)$ all of which have size at most $k$. Using BFS, we check that $\dm(s,v) \leq D$ for all $s \in S, v \in V$. This step ensures that, for every vertex with a ``sufficiently large'' (meaning, size-$k$) distance-$D/2$ out-neighborhood or distance-$\lceil D/2 \rceil$ in-neighborhood, some vertex in $S$ lies in that out- or in-neighborhood. In this sense, $S$ covers all of the large neighborhoods, which will be useful later.

We partition $V$ into $2\theta = n/k^2$  closed intervals, $V_1, \dots V_{2\theta}$. We let $L = V_1 \cup \cdots \cup V_{\theta}$ and $R = V_{\theta+1} \cup \cdots \cup V_{2\theta}$ be the left and right halves of $V$. We start in the middle and work outwards verifying that min-distances between vertices in $L$ and $R$ are at most $\lceil 3D/2 \rceil$.

At each inductive step, we consider two subsets of vertices $A = V_i$ and $B = V_j$ to the left and right of the ``middle interval,'' which consists of the intervening subsets $V_{i+1} \cup \cdots \cup V_{j-1}$. Intuitively, one might picture the middle interval as a growing amoeba, which at each step engulfs one of $A$ or $B$. We will only add vertices to the middle interval once they have been ``checked,'' so vertices from $A$ that are added to the middle interval have distance at most $3D/2$ to everything in $B$, and vice-versa. Eventually, either the algorithm will detect a min-distance greater than $D$ and report that $\MinDiam{G} > D$, or the amoeba will have engulfed enough of the graph that it contains the entirety of one of the halves $L$ or $R$, meaning that all distances from vertices in $L$ to vertices in $R$ are at most $\lceil \frac{3}{2}D \rceil$.

In order to expand the middle interval we have to confirm that one of our two candidate subsets $A$ and $B$ has small distances to the opposite half of the graph. Performing a BFS from each vertex in $A$ or $B$ to confirm this directly would be too slow, so instead we present two subroutines to achieve this goal faster. Algorithm \ref{alg:all-pairs} checks that all min-distances between vertices $a \in A$ and $b \in B$ are at most $\lceil 3D/2 \rceil$. Algorithm \ref{alg:directional} checks that either all vertices in $A$ have min-distance at most $\frac{3}{2}D$ to all vertices to the right of $B$, or that a symmetric property holds for all vertices in $B$. These algorithms will either detect that the min-diameter is more than $D$ or allow us to add one of $A$ or $B$ to the middle interval.

We first give a pseudocode description of the main algorithm, Algorithm \ref{alg:Full-graph-min-distance-tester}, which will give context for the two subroutines, Algorithms \ref{alg:all-pairs} and \ref{alg:directional}, whose descriptions follow afterwards.

\begin{algorithm}[h]
    \KwIn{DAG $G = (V,E)$, diameter guess $D$, parameters $k = m^{\frac{\beta + 1}{5\beta + 3 - 2\alpha\beta}}$ and $\theta = n/2k^2$.}
    \KwOut{One of the following. Each output verifies a corresponding property of $G$.
    \begin{tabular}{ccc}
		\midrule
		\Pass  & $\Rightarrow$ & $\MinDiam(G) \le \lceil \frac 32 D \rceil$  \\
		\Fail & $\Rightarrow$ & $\MinDiam(G) > D$\\
	\end{tabular}}
    
    Topologically sort $G$\;
    Using Lemma \ref{nbrhds} compute a $(k,(D/2, \lceil D/2 \rceil))$-neighborhood cover $S \subseteq V$. Run BFS to and from each vertex in $S$\; 
    \If{ $\exists s \in S$ such that $\epsilon(s) > D$}{\Fail}
    Partition $V$ into consecutive closed intervals, $V_1, \dots V_{2\theta}$, with $|V_i| = k^2$ for each $i$\;
    Initialize $i = \theta$ and $j = \theta + 1$\;
    \While{$i \ge 1 \andt j \le 2\theta$}{
        Run \cref{alg:all-pairs} (all-pairs) on the pair $(V_i, V_j)$. If this fails, then \Fail\;
        Run \cref{alg:directional} (directional tester) on the pair $(V_i, V_j)$. If this fails, then \Fail\;
        \uElseIf{\cref{alg:directional} passes and returns $V_i$}{$i = i - 1$\;
        }
        \Else{
        \tcc{Otherwise, \cref{alg:directional} passes and returns $V_j$.}
        $j = j + 1$\;\label{line:end-while-loop}
        }
    }
    \tcc{If this line is reached, all distances from $L$ to $R$ are at most $\lceil \frac 32 D \rceil$.}
    Recursively call this algorithm on $G[L]$ and $G[R]$. If either fails, \Fail. Else \Pass\;
        
\caption{Full graph min-distance tester}\label{alg:Full-graph-min-distance-tester}
\end{algorithm}

We will present the correctness and runtime analyses for the two subroutines, and then present the correctness proof and runtime analysis for the overall algorithm.

The first subroutine, Algorithm \ref{alg:all-pairs}, checks distances between pairs of vertices in subsets $A, B$ by using the cover $S$ to hit large neighborhoods and using fast sparse rectangular matrix multiplication to efficiently check set intersections between small neighborhoods. 

\begin{theorem}[\cite{rectangular}]\label{rectangular}
If $M, M'$ are $p \times l$ matrices having at most $l$ nonzero entries, where $p^{1+\frac{\alpha}{2}} \leq l \leq p^{\frac{\omega+1}{2}}$, then in time $O(l^{\frac{2\beta}{\beta+1}}p^\frac{2-\alpha\beta}{\beta+1})$ one can compute the product $M^T M'$.\footnote{To be precise, if we have $\alpha \geq a, \omega \leq c$, we can define $b = \frac{c - 2}{1 - a}$, and then the theorem holds for any such pair $a, b$, used in place of $\alpha, \beta$.}
\end{theorem}

\begin{algorithm}[h]
    \KwIn{DAG $G = (V,E)$, topological ordering $\pi$, subsets $A <_\pi B \subseteq V$ with $|A|=|B|=p=O(k^2)$, diameter guess $D$, $(k, (D/2, \lceil D/2 \rceil))$-neighborhood cover $S \subseteq V$ such that $\epsilon(s) \leq D$ for all $s \in S$, and $N^{out}_{D/2,S}(a)$, $N^{in}_{\lceil D/2 \rceil, S}(b)$ for $a \in A, b \in B$.}
    \KwOut{One of the following. Each output verifies a corresponding property of $G$.
    \begin{tabular}{ccc}
		\midrule
		\Pass  & $\Rightarrow$ & $\dm(a,b) \le \lceil \frac 32 D \rceil $ for all in $a \in A$ and $b \in B$  \\
		\Fail & $\Rightarrow$ & $\MinDiam(G) > D$\\
	\end{tabular}}
    Compute $M_A$, the matrix with columns given by indicator vectors of $N^{out}_{D/2, S}(a)$ for $a \in A$\;
    Compute $M_B$, the matrix with columns given by indicator vectors $N^{in}_{\lceil D/2 \rceil, S}(b)$ for $b \in B$\;
    Compute $M=M_A^TM_B$\;
    \ForEach{$a \in A, b \in B$}{ 
        \uIf{ $M_{ab} \ge 1$\label{line:cond-0}}
        {\Continue\;}
        \uElseIf{$a <_\pi N^{in}_{\lceil D/2 \rceil, S}(b)$ \andt
        $N^{in}_{\lceil D/2 \rceil, S}(b) \cap S \neq \emptyset$\label{line:cond-1}}
        {\Continue\;}
        \uElseIf{$b >_\pi N^{out}_{D/2,S}(a)$ \andt $N^{out}_{D/2,S}(a) \cap S \neq \emptyset$\label{line:cond-2}} 
        {\Continue\;}
        \Else{\Fail\;}
    }
    \Pass;
\caption{All-pairs min-distance tester}\label{alg:all-pairs}
\end{algorithm}

\begin{lemma}\label{lemma:all-pairs}
\cref{alg:all-pairs} produces the correct output in runtime $O(k^{4 + \frac{2\beta - 2\alpha\beta}{\beta + 1}})$.
\end{lemma}
\begin{proof}
We first show correctness. Assume the algorithm fails; then it fails inside the \textbf{foreach} loop at some pair $a,b$. Suppose for the sake of contradiction that $\dm(a,b) \leq D$. Let $x$ be a midpoint of the shortest path from $a$ to $b$, so $d(a,x) \leq D/2$, $d(x,b) \le \lceil D/2 \rceil$. If $x \notin N^{out}_{D/2, S}(a)$ despite being distance $\le D/2$ from $a$, then $x$ must lie after the out-neighborhood of $a$ reaches the $(k, (D/2, \lceil D/2 \rceil))$-neighborhood cover $S$ and is cut off. That is, there is some $s \in N^{out}_{D/2}(a) \cap S$  to the left of $x$, and hence to the left of $b$. Thus the condition in line \ref{line:cond-2} holds, and we \textbf{continue} rather than \textbf{FAIL}, which is a contradiction. Similarly, if $x \notin N^{in}_{\lceil D/2 \rceil, S}(b)$ then the condition in line \ref{line:cond-1} holds, which is a contradiction. We conclude that $x \in N^{out}_{D/2, S}(a) \cap N^{in}_{\lceil D/2 \rceil, S}(b)$, in which case $M_{ab} \ge 1$, completing the contradiction. Thus, $\dm(a,b) > D$, and so $\MinDiam(G) > D$.

Conversely, assume the algorithm passes. Then for each pair of vertices $a \in A, b \in B$, one of the conditions in lines \ref{line:cond-0}, \ref{line:cond-1}, or \ref{line:cond-2} must hold. If $M_{ab} \geq 1$, then there is some $x \in N^{out}_{D/2, S}(a) \cap N^{in}_{\lceil D/2 \rceil, S}(b)$, giving $d(a,b) \le d(a,x) + d(x,b) \le D/2+\lceil D/2 \rceil \le D+1$. Otherwise if $a <_\pi N^{in}_{\lceil D/2 \rceil, S}(b)$ and $N^{in}_{\lceil D/2 \rceil, S}(b) \cap S$ is nonempty, then there is some $s \in N^{in}_{\lceil D/2 \rceil, S}(b) \cap S$ with $s \geq_\pi a$. Thus $d(a,b) \le d(a,s) + d(s,b) \le D + \lceil D/2 \rceil \le \lceil 3D/2 \rceil$. Similarly, if $b >_\pi N^{out}_{D/2, S}(a)$ and $N^{out}_{D/2, S}(a) \cap S$ is nonempty, then there is some $s \in N^{out}_{D/2}(a) \cap S$ with $s \leq_\pi b$. So $d(a,b) \le d(a,s) + d(s,b) \le D/2 + D \le 3D/2$. We conclude that if the algorithm passes, then every pair $a \in A, b \in B$ satisfies $d(a,b) \le \lceil 3D/2 \rceil$.

We conclude with runtime analysis. $M_a$ and $M_b$ are sparse $p \times n$ matrices with $O(pk) = O(k^3)$ entries, so may be treated as $p \times pk$ matrices. Using Theorem \ref{rectangular}, the matrix multiplication takes time at most $O((pk)^{\frac{2\beta}{\beta+1}}p^\frac{2-\alpha\beta}{\beta+1}) = O(k^{4 + \frac{2\beta - 2\alpha\beta}{\beta + 1}})$. The \textbf{foreach} loop takes $O(|A||B|) =O(k^4)$ time. 
\end{proof}

\begin{remark}
One can modify Algorithm \ref{alg:all-pairs} by computing set intersections of the sets $N^{out}_{D/2, S}(a)$, $N^{in}_{\lceil D/2 \rceil, S}(b)$ for $a \in A, b\in B$, by brute force in lieu of matrix multiplication. This gives a combinatorial version of Algorithm \ref{alg:all-pairs} which runs in time $\tilde{O}(p^2k) = \tilde{O}(k^5)$.
\end{remark}

Next we give the directional min-distance tester. A call to this subroutine will prove a min-distance bound of $3D/2$ from $A$ to everything past $B$, prove a bound of $\lceil 3D/2 \rceil$ from $B$ to everything before $A$, or find a pair of vertices with min-distance greater than $D$. The idea is to iterate over vertices $a$ in $A$; those with large $D/2$-neighborhoods get hit by $S$, and if some $a$ has a small $D/2$-neighborhood it can be used as a jumping-off set to show $B$ is close to everything past $A$.

\begin{algorithm}
    \KwIn{DAG $G = (V,E)$, topological ordering $\pi$, closed subsets $A <_\pi B \subseteq V$, diameter guess $D$, parameter $k$, $(k, (D/2, \lceil D/2 \rceil))$-neighborhood cover $S \subseteq V$ with $\epsilon(s) \leq D$ for all $s \in S$, and neighborhoods $N^{out}_{D/2,S}(a)$ for all $a \in A$.}
    \KwOut{One of the following. Each output verifies a corresponding property of $G$.
    \begin{tabular}{ccc}
		\midrule
		\Pass and return $A$ & $\Rightarrow$ & $d(a,v)\le\frac 32D$ for all $a \in A$ and $v >_\pi B$  \\
		\Pass and return $B$ & $\Rightarrow$ & $d(v,b)\le \lceil \frac 32D \rceil $ for all $b \in B$ and $v <_\pi A$ \\
		\Fail & $\Rightarrow$ & $\MinDiam(G) > D$\\
	\end{tabular}}
    \ForEach{$a \in A$}{
        \If{ $N^{out}_{D/2, S}(a) \cap S = \emptyset$ or $N^{out}_{D/2, S}(a) \cap S >_\pi B$\label{line:big-out-nbd}}{
            \ForEach{ $v \in \{a\} \cup N^{out}_{D/2, S}(a)$ }{
                BFS to and from $v$\; 
                \If{$\epsilon(v) > D$\label{line:fail-cond-2}}
                {\Fail\;}
            }
            \Pass and \Return $B$\;
        }       
    }
\Pass and \Return $A$\;
\caption{Directional min-distance tester}\label{alg:directional}
\end{algorithm}

\begin{lemma}\label{lemma:directional}
Algorithm \ref{alg:directional} produces the correct output in runtime $O(mk)$.
\end{lemma}
\begin{proof}
We first show correctness. If the algorithm fails, we must have found a vertex $u$ with $\epsilon(u) > D$, and thus $\MinDiam(G) > D$.

If the algorithm returns $A$, then for each $a \in A$, there is some $s \in S \cap N^{out}_{D/2, S}(a)$ with $s$ appearing in $B$ or to its left. Thus for all $u >_\pi B$, we have $d(a,u) \le d(a,s)+d(s,u) \le D/2+D =3D/2$.

Otherwise, the algorithm returns $B$. Then the condition in line \ref{line:big-out-nbd} must hold for some $a \in A$. Let $u <_\pi A$ and $b \in B$. Since we do not fail in line \ref{line:fail-cond-2}, we must have $d(a,b) \le D$. Let $x$ be a midpoint of the shortest path from $a$ to $b$, so $d(a,x) \leq D/2$ and $d(x,b) \leq \lceil D/2 \rceil$. Since $N^{out}_{D/2,S}(a)$ is either not cut off by hitting $S$ or is cut off after $B$, we have $x \in N^{out}_{D/2,S}(a)$. Finally, since we do not fail in line \ref{line:fail-cond-2}, we must have $d(u,x) \le D$. Concluding, $d(u,b) \le d(u,x)+d(x,b) \le D + \lceil D/2 \rceil \le \lceil 3D/2 \rceil$. As $u,b$ were arbitrary, this completes the case.

We conclude with runtime analysis. The outer \textbf{foreach} loop repeats until we have covered every vertex in $A$ or the condition in line \ref{line:big-out-nbd} is satisfied. Checking this condition takes time at most $O(k)$ for a total of $ O(|A|k) = O(mk)$. If the condition is satisfied, we perform $1+|N^{out}_{D/2,S}(a)|$ calls to BFS, for a total of $O(mk)$. The algorithm concludes before returning to the outer loop, so we may add these contributions for a total time of $O(mk).$
\end{proof}

We can now complete the analysis of the overall algorithm.

\begin{lemma}\label{lemma:full-graph-min-distance-tester}
Algorithm \ref{alg:Full-graph-min-distance-tester} produces the correct output in runtime $\tilde{O}(m^{\frac{4\beta + 2 - 2\alpha\beta}{5\beta+3-2\alpha\beta}}n)$. 
\end{lemma}
\begin{proof}
This algorithm fails only when some $s \in S$ has $\epsilon(s) > D$, when \cref{alg:all-pairs} fails, or \cref{alg:directional} fails, all of which imply $\MinDiam(G) > D$. We now show that in the event of a pass, $\MinDiam(G) \le \lceil 3D/2 \rceil$. It suffices to prove that if the algorithm reaches line \ref{line:end-while-loop} then all min-distances between vertices in $L$ and $R$ are at most $\lceil 3D/2 \rceil$. 

Assume we are at the beginning of iteration $t$ of the \textbf{while} loop. Let $I_t = \bigcup_{i < \ell < j} V_\ell$ be the interval strictly between $V_i$ and $V_j$. We will show inductively that for all $x \in L$ and $y \in R$, where at least one of $x$ or $y$ is in the interval $I_t$, $\dm(x,y) \leq \lceil 3D/2 \rceil$. When the loop terminates, $I_t$ must entirely contain either $L$ or $R$, which will prove that all min-distance between vertices in $L$ and $R$ are at most $\lceil 3D/2 \rceil$.

The base case is trivial, as $I_1$ is empty. Assume the claim holds for $t$. Without loss of generality, assume \cref{alg:directional} returns $V_i$, so that $I_{t+1} = I_t \cup V_i$. Let $x \in L, y \in R$ with at least one of $x$ or $y$ in $I_{t+1}$. If $x$ or $y$ is in $I_t$, then $d(x,y) \le \lceil 3D/2 \rceil$ by induction. Otherwise, one must lie in $I_{t+1}\setminus I_t = V_i$, and since $V_i \subset L$, this vertex must be $x$. Then we have $y \in V_j$ or $y>_\pi V_j$. If $y \in V_j$, then since \cref{alg:all-pairs} did not fail, we have $d(x,y) \le 3D/2$. If $y >_\pi V_j$, then because \cref{alg:directional} returned $V_i$ we have $d(x,y) \le \lceil 3D/2 \rceil$. This completes the induction and the proof of correctness.

$G$ can be topologically sorted in time $\tilde{O}(n)$. Lemma \ref{nbrhds} constructs the set $S$ in time $O(nk^2)$. Running BFS to and from each vertex in $S$ takes time $\tilde{O}(mn/k)$. We run Algorithm \ref{alg:all-pairs} and Algorithm \ref{alg:directional} each up to $2\theta = n/2k^2$ times. Since Algorithm \ref{alg:all-pairs} takes time $O(k^{4 + \frac{2\beta - 2\alpha\beta}{\beta + 1}})$ and Algorithm \ref{alg:directional} takes time $O(mk)$, the total runtime of a recursive step is therefore:
$$\tilde{O}(mn/k + nk^{2 + \frac{2\beta - 2\alpha\beta}{\beta + 1}})$$

Setting $k = m^{\frac{\beta + 1}{5\beta + 3 - 2\alpha\beta}}$, we obtain $\tilde{O}(m^{\frac{4\beta + 2 - 2\alpha\beta}{5\beta+3-2\alpha\beta}}n)$ for each recursive step. The recursion over the left and right halves of the graph then adds a logarithmic factor.
\end{proof}

\begin{proof}[Proof of \cref{thm:3/2-unweighted-dag}]
Binary searching over $D \in [n]$ using Algorithm \ref{alg:Full-graph-min-distance-tester} gives the desired algorithm. The additive +1/2 follows from the fact that $\lceil 3D/2 \rceil$ may equal $3D/2 + 1/2$.
\end{proof}

\begin{remark} Using the combinatorial version of Algorithm \ref{alg:all-pairs}, one can obtain a combinatorial $(\frac{3}{2}, \frac{1}{2})$-approximation algorithm for min-diameter which runs in time $\tilde{O}(m^{3/4}n)$.
\end{remark}

\begin{theorem}\label{thm:exact}
There is a $\frac{3}{2}$-approximation algorithm for min-diameter in unweighted DAGs that runs in time $\tilde{O}(m^{\frac{8\beta+4-4\alpha\beta}{5\beta+3-2\alpha\beta}}n^{\frac{\beta+1}{5\beta+3-2\alpha\beta}})$.
\end{theorem}

This runtime is at most $O(m^{1.426}n^{0.288})$. We note that there is also a combinatorial version of this algorithm which runs in time $\tilde{O}(m^{5/4}n^{1/2})$, and a version which runs in $\tilde{O}(m^{\frac{7\beta + 3 -3\alpha\beta}{5\beta+3-2\alpha\beta}}n^{\frac{2\beta+2-\alpha\beta}{5\beta+3-2\alpha\beta}})$ time, which is $O(m^{1.171}n^{0.543})$, when $m \leq n^{1.283}$. A proof of the theorem may be found in Appendix \ref{app:3/2-mindiam}.

\section{Bichromatic min-diameter}

We first present OV-based conditional lower bounds for bichromatic min-diameter approximations, both for general graphs and for DAGs. Then in Section \ref{section:near2bichmindiam}, we give an almost-$2$-approximation algorithm for bichromatic min-diameter in DAGs.

\subsection{Conditional lower bounds for bichromatic min-diameter}

We present, in order of increasing strength of the bound, three conditional lower bounds for approximating bichromatic min-diameter: one for DAGs, one for unweighted directed graphs, and one of weighted directed graphs. The constructions proceed analogously to \cref{thm:2-approx-lb}.

\begin{theorem}\label{thm:bichr-LB-DAG}
For any $\epsilon, \delta > 0$ if there is an $O(n^{2-\epsilon})$-time algorithm giving a $(2-\delta)$-approximation for bichromatic min-diameter in unweighted DAGs with $O(n)$ vertices and $O(n^{1+o(1)})$ edges, then the OV Hypothesis is false.
\end{theorem}

\begin{proof}
Given an OV instance with vector sets $A, B$, for each integer $t > 0$, we will construct a bichromatic DAG $G_t$ whose bichromatic min-diameter $D$ is $2t+1$ if there are a pair of orthogonal vectors $a \in A, b\in B$ and at most $t+1$ otherwise.

$G_t$ has red vertex set $A$ and blue vertex sets $B_1, \dots B_t$, where each $B_i$ has vertices corresponding to the vertices of $B$. Note $|A| = |B_i| = n$ for all $i \in [t]$. $G$ also has blue vertex sets $I_1, \dots I_t$ of size $O(\log n)$, where each vertex $v_{i,j}$ in $I_i$ corresponds to an index $j$ of a vector.

The edges of $G_t$ are as follows: $A$ is complete to $I_1$. For $i \in [t-1]$, there are perfect matchings between $I_i, I_{i+1}$ and $B_i \to B_{i+1}$, consisting respectively of edges $(v_{i, j} \to v_{i+1, j})$ and of edges $(b_i, b_{i+1})$ where $b_i \in B_i, b_{i+1} \in B_{i+1}$ correspond to the same vector $b \in B$. Lastly, there are edges $(a, v_{t, j})$ for each index $j$ such that $a[j] = 1$ and edges $(v_{t, j}, b_1)$ for each index $j$ such that $b[j] = 1$.

\begin{figure}[h]
    \centering
\includegraphics[scale=0.27]{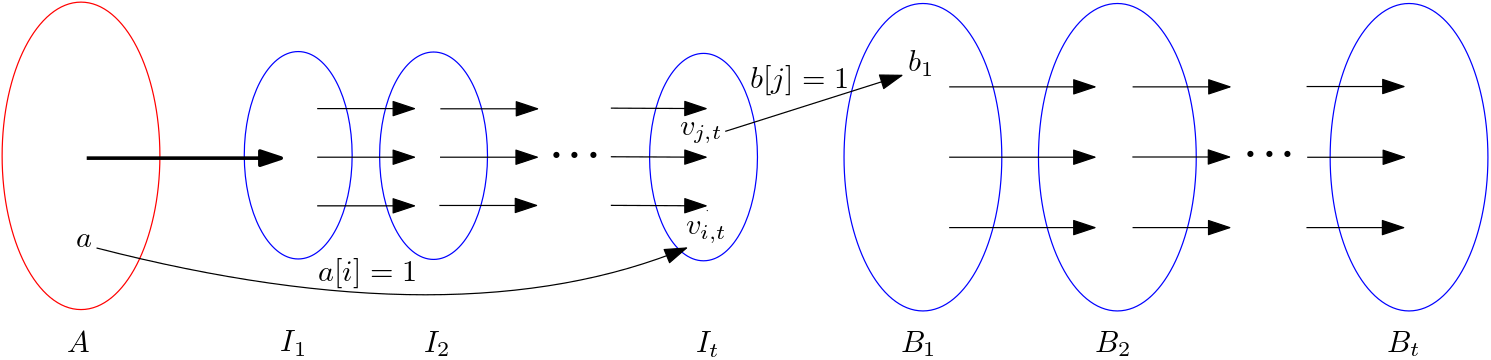}
    \caption{The graph $G_t$. The thick edge denotes that all possible edges $A \to I_1$ exist.}
\label{fig:bich-DAG-lb}
\end{figure}
\medskip

\textbf{YES case:} Suppose $a \in A, b \in B$ are orthogonal vectors. Since the edges between sets $B_i$ are matchings, any path $P$ from $a$ to $b_t$ must pass through $b_0$. The vertex prior to $b_0$ in $P$ must be $v_{t, j}$ for some index $j$ such that $b[j] = 1$. Then we must have $a[j] = 0$, since the vectors are orthogonal, so there is no edge from $a$ to $v_{t, j}$. Thus, the path from $p$ to $v_{t, j}$ must pass through the sets $I_i$. Hence, we have $d(a,b_t) = 2t+1$, giving a bichromatic min-diameter of at least $2t+1$.

\textbf{NO case:} Suppose that there is no pair of orthogonal vectors. Then for every pair $a \in A$, $b \in B$, there is some index $j$ such that $a[j] = b[j] = 1$. Hence for all $i \in [t]$, there is a path $a \to v_{t, j} \to b_1 \to \cdots \to b_i$, which has length at most $t+1$. Thus, for all $a \in A$ and $b_i \in B_i$ for any $i \in [t]$, $d(a,b_i) \leq t+1$. Since $A$ is complete to $I_1$, for all $a \in A$ and $v_{i, j} \in I_i$, there is a path $a \to v_{1, j} \to \cdots \to v_{i,j}$, which has length at most $t$. Hence, every red vertex is at distance at most $t+1$ from every blue vertex, giving a bichromatic min-diameter of at most $t+1$.

For constant $t$, constructing $V(G_t)$ takes $O(tn) = O(n)$ time, and constructing $E(G_t)$ takes $O(tn\log n)$ which is $O(n\log n)$ time. So if we can obtain a better than $\frac{2t}{t+1}$-approximation for the bichromatic min-diameter of $G_t$ in $O(n^{2-\epsilon})$ time, we have also solved the OV instance in $O(n^{2-\epsilon})$ time.

\end{proof}

\begin{theorem}\label{thm:LB-bichr-unw}
For any $\epsilon, \delta > 0$, if there is an $O(n^{2-\epsilon})$-time algorithm giving a $(5/2-\delta)$-approximation for bichromatic min-diameter in unweighted graphs with $O(n)$ vertices and $O(n^{1+o(1)})$ edges, then the OV Hypothesis is false.
\end{theorem}

\begin{proof}
Given an OV instance with vector sets $A, B$, for each integer $t > 0$, we will construct a graph $G$ whose bichromatic min-diameter $D$ is at least $5$ if there are a pair of orthogonal vectors $a \in A, b\in B$ and at most $2$ otherwise.

$G$ has red vertex set $A$ and blue vertex set $B$ corresponding to the vector sets $A,B$, a blue vertex set $I$ corresponding to the set of indices $I$ of vectors in $A, B$, and an auxiliary blue vertex $x$.

For $i \in I$ and $a \in A$, we add an edge $(a,i)$ iff $a[i] = 1$. For $i \in I$ and $b \in B$, we add an edge $(i, b)$ iff $b[i] = 1$. Finally, we add all edges from $I$ to $x$ and from $x$ to $A$. 

We now compute the bichromatic min-diameter of $G$. The only red vertices are in $A$, so we examine their distances. Each red vertex is adjacent to $x$ and therefore is distance at most 2 from every vertex in $I$. Fix $a \in A$ and $b \in B$. As $b$ has outdegree 0, any path connecting them must start at $a$. The only outneighbors of $a$ are vertices $i \in I$ such that $a[i] = 1$; if $a$ and $b$ are not orthogonal then at least one such $i$ satisfies $b[i] = 1$ and therefore contains an edge to $b$. So when no $a,b$ are orthogonal the min-diameter is indeed 2. On the other hand, for $a,b$ orthogonal, there is no such $i$, so any path from $a$ to $b$ must begin $a \to i \to x$ (and continue $x \to a' \to i'$) so we cannot reach $b$ in four or fewer edges, and the min-diameter is at least 5.

\end{proof}

\begin{theorem}\label{thm:bichr-LB-w}
For any $\epsilon, \delta > 0$, if there is an $O(n^{2-\epsilon})$-time algorithm giving a $(3-\delta)$-approximation for bichromatic min-diameter in weighted graphs with $O(n)$ vertices and $O(n^{1+o(1)})$ edges, then the OV Hypothesis is false.
\end{theorem}

\begin{proof}
Given an OV instance with vector sets $A, B$, for each integer $t > 0$, we will construct a graph $G$ whose bichromatic min-diameter $D$ is at least $3t+2$ if there are a pair of orthogonal vectors $a \in A, b\in B$ and at most $t+1$ otherwise.

$G$ has red vertex set $A$ and blue vertex set $B$ corresponding to the vector sets $A,B$, as well as blue vertex set $I$ corresponding to the set of indices $I$ of vectors in $A, B$.

For $i \in I$ and $a \in A$, we add a weight-$t$ edge $(a,i)$ iff $a[i] = 1$. For $i \in I$ and $b \in B$, we add a weight-1 edge $(i, b)$ iff $b[i] = 1$. Lastly, we add all possible edges $(i,a)$ and give these ``backwards'' edges weight $t+1$.

For any $a \in A, i \in I$, $\dm(a,i) \leq t+1$ since the edge $(i,a)$ exists and has weight 1. If $b \in B$ is not orthogonal to $a \in A$, then there exists $i \in I$ such that $a[i] = b[i] = 1$, so there is a path from $a$ to $b$ via $i$. Thus $\dm(a,b) \leq t+1$. So if there is no pair of orthogonal vectors, the bichromatic min-diameter is at most $t+1$.

Suppose that $a,b$ are orthogonal vectors. Consider a shortest $a\to b$ path; since all edges out of $a$ are into $I$ and all edges into $b$ are out of $I$, the path must contain some $i, i' \in I$ such that $a[i] = 1$ and $b[i'] = 1$. We have $i \neq i'$ since $a, b$ are orthogonal. The shortest path between distinct $i, i'$ is a path of the form $i \to a' \to a'$ where $a' \in A$ is some vertex such that $a'[i'] = 1$. Thus, a shortest $a \to b$ path must be of the form $a, i, a', i', b$ and must therefore have weight $3t+2$, since $w(a, i) = w(a', i') = t$, $w(i, a') = t+1$, and $w(i', b) = 1$. So in this case the bichromatic min-diameter is at least $3t+2$.

\end{proof}

\subsection{Almost-2-approximation for bichromatic min-diameter in DAGs}\label{section:near2bichmindiam}

We conclude by turning to upper bounds for bichromatic min-diameter. 

\begin{theorem}\label{thm:full-bix-dag-algo}
There is an $\tilde{O}(\min(m^{4/3}n^{1/3}, m^{1/2}n^{3/2}))$-time algorithm, which, given a DAG $G$ with maximum red-blue edge weight $M_0$, outputs an approximation $D_0$ such that $D \leq D_0 < 2D + M_0 \leq 3D$.
\end{theorem}

We first show how to efficiently achieve a 2-approximation in the case of separated DAGs. Recall that these are topologically ordered DAGs with color sets $A, B$ such that $A <_\pi B$.

The following subroutine will allow us to identify vertices with relatively few out-neighbors of the same color, and verify that the rest of the vertices of that color have small distances to vertices of the opposite color. 
\begin{lemma}\label{large}
Let $G$ be an $(A,B)$-separated DAG. Given parameters $D$ and $k = n^{\epsilon}$ for $\epsilon \in [0,1]$, one can in time $O(\frac{mn}{k}\log^2 n)$  either determine that $\BichromaticMinDiam{G} > D$ or compute a set $A' \subseteq A$ such that for all $a \in A'$, there are at most $k$ vertices $v \in A$ such that $d(a,v) \leq D$, and for all $a \not\in A', b \in B$, $\dm(a,b) \leq 2D$. 
\end{lemma}

\begin{proof}
Without loss of generality, we assume $A$ is red and $B$ is blue.

Using Lemma \ref{nbrhds}, we find a $(k, D)$-neighborhood cover $S$ of size $O(\frac{n}{k}\log n)$ along with the sets $N^{out}_{D, S}(v)$. This takes time $O(nk^2)$. Note that $|N^{out}_{D, S}(a) \cap A| \leq k$ for all $a \in A$.

We run BFS to and from all vertices in $S$, which takes time $O(\frac{mn}{k}\log^2 n)$. If we find an $s \in S, v \in V$ of different colors such that $\dm(s,v) > D$, then $\BichromaticMinDiam{G} > D$. Otherwise, we iterate over vertices $a \in A$. Suppose that $N^{out}_{D, S}(a) \cap A$ does not intersect $S$; this can be checked in $O(k)$ time for each $a \in A$. Since $N^{out}_{D, S}(a) \cap A$ is the set of red vertices at distance at most $D$ from $a$ cut off after the first time this set intersects $S$, this means that $N^{out}_{D, S}(a) \cap A$ contains all red vertices at distance at most $D$ from $a$. In this case, we add $a$ to the set $A'$. Otherwise, $N^{out}_{D, S}(a) \cap A$ intersects $S$ at some vertex $s$, which must be red. Then we have $d(a, s) \leq D$ and $d(s, b) \leq D$ for all blue vertices $b$. So $d(a,b) \leq 2D$ for all $b \in B$.
\end{proof}

We can now present a  2-approximation algorithm for the special case of separated DAGs.

\begin{algorithm}[h]
    \KwIn{$(A,B)$ separated DAG $G = (A \sqcup B,E)$, where (without loss of generality) $A$ is red and $B$ is blue. Diameter guess $D$ and parameters $k, \Delta$.}
    \KwOut{One of the following. Each output verifies a corresponding property of $G$.
    \begin{tabular}{ccc}
		\midrule
		\Pass  & $\Rightarrow$ & $\BichromaticMinDiam(G) \le  2D$  \\
		\Fail & $\Rightarrow$ & $\BichromaticMinDiam(G) > D$\\
	\end{tabular}}
	Using Lemma \ref{large}, either determine that $\BichromaticMinDiam{G} > D$ and \Fail, or compute a set $A' \subseteq A$ such that for all $a \in A'$, there are at most $k$ red vertices at distance $D$ from $a$, and for all $a \in A \setminus A'$, $\dm(a,b) \leq 2D$ for all $b \in B$\;
    Let $T$ be the set of red vertices with out-degree at least $\Delta$. Run Dijkstra's algorithm out of every vertex in $T$. If we find a $t \in T, b \in B$ such that $d(t,b) > D$, then \Fail\;
    \ForEach{$a$ in $A'$}{
    Compute the set $N^{out}_D(a) \cap A$\;
    \uIf{$N^{out}_D(a) \cap T = \emptyset$}{
    Run Dijkstra's algorithm from $a$. If $d(a,b) > D$ for some $b \in B$, then \Fail\;
    Run Dijkstra's algorithm into each blue vertex adjacent to a vertex in $N^{out}_D(a) \cap A$; i.e., each blue vertex in $N(N^{out}_D(a) \cap A)$. If we find a red-blue distance greater than $D$, then \Fail\;
    \Pass\; \label{line:earlier-pass}
    }
    \Else{\Continue\;}
    }
    \Pass\; \label{line:final-pass}

\caption{Bichromatic min-distance tester for sparse separated DAGs}\label{alg:sparse-tricky}
\end{algorithm}

\begin{lemma}\label{trickysparse}
\cref{alg:sparse-tricky} provides an $O(m^{4/3}n^{1/3}\log^2 n)$-time algorithm which, given an $(A,B)$-separated DAG $G$, together with a parameter $D$, determines either that   $\BichromaticMinDiam{G} > D$ or that  $\BichromaticMinDiam{G} \leq 2D$.
\end{lemma}

\begin{proof}
We assume $A$ is red and $B$ is blue. If the algorithm returns \Fail, then it found a red-blue pair of vertices $a, b$ such that $\dm(a,b) > D$, so $\BichromaticMinDiam{G} > D$.

Otherwise, suppose the algorithm returns \Pass. Then the application of Lemma \ref{large} successfully found a set $A'$ with the desired properties. In particular, for all $a \in A \setminus A'$ and all $b \in B$, $\dm(a,b) \leq 2D$.

If the algorithm returns \Pass at line \ref{line:final-pass}, then we have for all $a \in A'$ that $N^{out}_D(a) \cap A$ intersects $T$ at some vertex $t$, which is necessarily red. Then  $d(a, t) \leq D$ and $d(t, b) \leq D$ for all $a \in A'$ and $b \in B$, so (combined with the fact that this already holds for $a \in A \setminus A')$, we have that $d(a, b) \leq 2D$ for all red-blue pairs $a, b$, so $\BichromaticMinDiam{G} \leq 2D$.

Otherwise, the algorithm returns \Pass at line \ref{line:earlier-pass}. So there is some $a \in A'$ such that $N^{out}_D(a) \cap A$ does not intersect $T$. In this case we run Dijkstra's algorithm from all blue vertices in $N(N^{out}_D(a) \cap A)$, and from $a$ itself. Let $b \in B$ be any blue vertex. We note that on the shortest path $P$ between $a$ and $b$, there is a left-most blue vertex, $t$. Then since $d(a,t) \leq D$, the vertex $t'$ appearing before $t$ on the path $P$ is red, and $d(a, t') \leq D$. Therefore $t' \in N^{out}_D(a) \cap A$, so $t \in N(N^{out}_D(a) \cap A)$. Then via Dijkstra, we have that for all $a' \in A$, $d(a', t) \leq D$, and so $d(a', b) \leq 2D$. Since this holds for all $b \in B$, we again have $\BichromaticMinDiam{G} \leq 2D$, and have completed the proof of correctness.

We note that \cref{large} has runtime $O(\frac{mn}{k}\log^2 n)$. We further note that $T$ has size $O(m/\Delta)$, so running Dijkstra from all vertices in $T$ takes time $O(\frac{m^2}{\Delta} \log n)$. Moreover, we note that for $a \in A'$, the set $N^{out}_D(a) \cap A$ has size at most size $k$, so computing it takes time at most $O(k^2)$. Thus, computing these sets takes time at most $O(|A'|k^2) \leq O(nk^2)$. Finally, note that if $N^{out}_D(a) \cap T = \varnothing$, then all vertices in $N^{out}_D(a) \cap A$ have degree less than $\Delta$. So computing $N(N^{out}_D(a) \cap A)$ takes time at most $O(\Delta|N^{out}_D(a) \cap A|) \le O(\Delta k)$, and its size is also bounded by $\Delta k$. Running Dijkstra from $a$ and from all blue vertices in $N(N^{out}_D(a) \cap A)$ therefore takes time $O(m\Delta k \log n)$. Fix $k = n^{2/3}m^{-1/3}$ and $\Delta = m^{2/3}{n^{-1/3}}$. Our total runtime is $O(\frac{m^2}{\Delta} \log n + nk^2 + m\Delta k \log n + \frac{mn}{k}\log^2 n) = O(m^{4/3}n^{1/3}\log^2 n)$ as claimed.
\end{proof}

A similar algorithm gives a better runtime in the case where the graph is dense $(m > n^{7/5})$:

\begin{lemma}\label{trickydense}
\cref{alg:dense-tricky} provides an $O(m^{1/2}n^{3/2}\log^2 n)$-time algorithm which, given an $(A,B)$-separated DAG $G$, together with a parameter $D$, determines either that $\BichromaticMinDiam{G} > D$ or that $\BichromaticMinDiam{G} \leq 2D$. 
\end{lemma}

\begin{algorithm}[h]
    \KwIn{DAG $G = (A \sqcup B,E)$, where (without loss of generality) $A$ is the set of red vertices, $B$ is the set of blue vertices, and $A <_\pi B$. Diameter guess $D$ and parameter $k$.}
    \KwOut{One of the following. Each output verifies a corresponding property of $G$.
    \begin{tabular}{ccc}
		\midrule
		\Pass  & $\Rightarrow$ & $\BichromaticMinDiam(G) \le  2D$  \\
		\Fail & $\Rightarrow$ & $\BichromaticMinDiam(G) > D$\\
	\end{tabular}}
	Using \cref{large}, either determine that $\BichromaticMinDiam{G} > D$ and \Fail, or compute a set $A' \subseteq A$ such that for all $a \in A'$, there are at most $k$ red vertices at distance $D$ from $a$, and for all $a \in A \setminus A'$, $\dm(a,b) \leq 2D$ for all $b \in B$\;
	Using \cref{large}, either determine that $\BichromaticMinDiam{G} > D$ and \Fail, or compute a set $B' \subseteq B$ such that for all $b \in B'$, there are at most $k$ blue vertices at distance $D$ from $b$, and for all $b \in B \setminus B'$, $\dm(a,b) \leq 2D$ for all $a \in A$\;
    \ForEach{$a$ in $A'$}{
    Compute the \textit{red out-neighborhood} $N^{out}_D(a) \cap A$ using Dijkstra's algorithm. Store all computed distances $d(a,v)$ for $v \in N^{out}_D(a) \cap A$.
    Initialize $BB(a) = \varnothing$\; 
    \ForEach{edge $(v, b)$ where $v \in N^{out}_D(a) \cap A, b \in B$}{
    Check if $d(a,b) = d(a,v) + w(v,b) \leq D$. If so, add $b$ to $BB(a)$.
    }
    \tcc{$BB(a) = \{b \in B \ | \ b \in N(N^{out}_D(a) \cap A), d(a,b) \leq D \}$ is the \textit{blue boundary} of the red out-neighborhood}
    }
    Run set intersection on the collections $\{BB(a) : a \in A'\}$ and $\{N^{in}_D(b) \cap B : b \in B'\}$. If $BB(a) \cap (N^{in}_D(b) \cap B) = \emptyset$ for some $a \in A'$ and $b \in B'$, then \Fail. Otherwise \Pass\;

\caption{Bichromatic min-distance tester for dense separated DAGs}\label{alg:dense-tricky}
\end{algorithm}

\begin{proof}
We assume $A$ is red and $B$ is blue. Applying \cref{large} twice takes time $O(\frac{mn}{k}\log^2 n)$. If the algorithm fails at this step, we have $\BichromaticMinDiam{G} > D$. Otherwise, it remains to check distances $\dm(a,b)$ where $a \in A'$ and $b \in B'$. Computing $N^{out}_D(a) \cap A$ for all $a \in A'$ and $N^{in}_D(b) \cap B$ for all $b \in B'$ takes time $O(nk^2)$, since these sets each have size at most $k$.
Computing $BB(a)$ requires iterating through every edge leaving $N^{out}_D(a) \cap A$. Since each vertex has $O(n)$ incident edges, this takes time  $O(n | N^{out}_D(a) \cap A|) \le O(nk)$. Repeating this for all $a \in A'$ therefore contributes $O(n^2k)$ to the runtime. 

Once we have computed $BB(a)$, we iterate through the vertices $b \in B'$. If $BB(a) \cap (N^{in}_D(b) \cap B) = \emptyset$ for some $a \in A'$ and $b \in B'$, then we claim $\dm(a,b) > D$. Suppose otherwise; let $y$ be the last red vertex on the shortest path from $a$ to $b$, and let $x$ be the first blue vertex. Clearly $\dm(a,x) \le D$ and so $x \in N^{out}_D(a) \cap A$. Then $y \in BB(a)$. However $\dm(y,b) \le D$ as well, so $y \in N^{in}_D(b) \cap B$, which is a contradiction. Thus if the algorithm fails,  $\dm(a,b) > D$.

Otherwise, for all $a \in A'$ and $B \in B'$ there is an $x \in BB(a) \cap (N^{in}_D(b) \cap B)$. Then by definition $d(a, x) \leq D$ and $d(x, b) \leq D$, so we have verified that $d(a,b) \leq 2D$ and can progress to checking the next blue vertex $b \in B'$. This completes the proof of correctness.

For the runtime, we have a contribution of $O(\frac{mn}{k}\log^2 n)$ from \cref{large} and a contribution of $O(n^2k)$ from the \textbf{foreach} loop. For set intersection it suffices to compute it combinatorially which takes time $O(n^2k\log k)$ as well. Taking $k = m^{1/2}n^{-1/2}$ we have a total runtime of $O(m^{1/2}n^{3/2}\log^2 n)$.
\end{proof}

We can now describe our algorithm for arbitrary bichromatic DAGs.

\begin{theorem}\label{bichdiam}
There is an $\tilde{O}(\min(m^{4/3}n^{1/3}, m^{1/2}n^{3/2}))$-time algorithm which, given a bichromatic DAG $G$ and a parameter $D$, determines whether $G$ has bichromatic min-diameter greater than $D$ or at most $2D + M$, where $M$ is the maximum edge weight in $G$.
\end{theorem}

The idea is to look at two overlapping closed subgraphs $H_0, H_1$ in the middle of the graph, where $H_0$ is $(A,B)$-separated and $H_1$ is $(B,C)$-separated for some $A,B,C \subseteq V$. Depending on the sparsity of the graph, we will use the algorithm of either Lemma \ref{trickysparse} or \ref{trickydense} to gain information about the distances within each $H_i$. We then run a constant number of Dijkstras, which will allow us to either find that $G$ has a bichromatic min-diameter greater than $D$, or find a partition $V = V_0 \cup V_1$ such that for all $v \in V_0, w \in V_1$ of different colors, $d(v, w) \leq 2D + M$. In the latter case, we recurse on $G[V_0]$ and $G[V_1]$.
\begin{proof}

We begin by topologically sorting the graph; let $v_1, \dots, v_n$ be the vertices of the graph in topologically sorted order $\pi$. We will choose a middle vertex $v_j$ as follows: If 
$m > n^{7/5}$ (so $m^{1/2}n^{3/2} > m^{4/3}n^{1/3}$),
i.e. if the graph is dense, let $j = \lceil \frac{n}{2} \rceil$. Otherwise, in $\tilde{O}(m)$ time we can find a vertex $v_j$ such that there are most $m/2$ edges with endpoints both to the right of $v_j$ and at most $m/2$ edges with endpoints both to the left of $v_j$. This can be done by iterating through vertices in sorted order and keeping a running count of how many edges are incident to those vertices, stopping when the count reaches $m/2$.

We will assume without loss of generality that $v_j$ is blue. Let $B$ be a maximal closed set of blue vertices containing $v_j$. Let $A$ and $C$ be maximal closed sets of red vertices such that $A <_\pi B <_\pi C$  and $A \cup B \cup C$ is closed. Let $H_0 = G[A \cup B]$, and let $H_1 = G[B \cup C]$.

Via Lemmas \ref{trickysparse} and \ref{trickydense}, we can in $\tilde{O}(\min(m^{4/3}n^{1/3}, m^{1/2}n^{3/2}))$ time check for each $H_i$ either that $H_i$ has bichromatic min-diameter greater than $D$, in which case we are done, or that for all differently colored pairs $(v, w) \in V(H_i)$, $d(v,w) \leq 2D$. We may assume the latter case holds for $i \in \{0, 1\}$.

Let $l$ be the right-most vertex to the left of $A$, and let $r$ be the left-most vertex to the right of $A$. Then $l$ and $r$ are blue, since $A$ and $C$ are maximal closed sets of red vertices. Let $a_0$ be the left-most vertex in $A$, and let $a_1$ be the right-most vertex in $A$; define $b_0, b_1, c_0, c_1$ correspondingly for $B$ and $C$. See Figure \ref{fig:bich-alg}.\footnote{Not all of these vertices need exist; if they do not, this implies there are no vertices beyond them in the relevant direction, so all of the distance checks that would have involved that vertex are unnecessary.}

\begin{figure}[h]
    \centering
\includegraphics[scale=0.45]{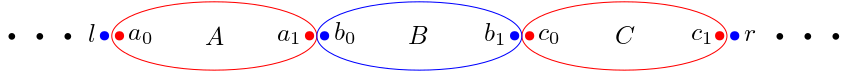}
\caption{}
\label{fig:bich-alg}
\end{figure}

Each pair $(l, a_0), (a_1, b_0), (b_1, c_0), (c_1, r)$ contains a red and a blue vertex which are consecutive under $\pi$. Check for each such pair that there is an edge between the vertices in the pair. If any such edge does not exist, we report that the bichromatic min-diameter is infinite.

We now run Dijkstra's algorithm into and out of each vertex $l, a_0, a_1, b_0, b_1, c_0, c_1, r$. If some red-blue distance found is larger than $D$, we report that the bichromatic min-diameter is greater than $D$. Otherwise, we will show that for any vertex $x$ in the set $V_0$ of vertices to the left of $B$, and any vertex $y \in V \setminus V_0$ having a different color than $x$, $d(x,y) \leq 2D+M$. The cases are as follows:

\begin{enumerate} 
\item Let $x <_\pi A$ be any red vertex, and let $y >_\pi A$ be any blue vertex. Then by assumption, we have checked that $d(x, l) \leq D, d(l, a_0) \leq M$, and $d(a_0, y) \leq D$, so we have verified that $d(x, y) \leq 2D+M$.

\item Let $x <_\pi A$ be any blue vertex, and let $y >_\pi A$ be any red vertex. Then by assumption, we have checked that $d(x, a_1) \leq D$, $d(a_1, b_0) \leq M$, and $d(b_0, y) \leq D$, so again we have verified that $d(x, y) \leq 2D+M$.

\item Let $x \in A$ be any vertex; note $x$ must be red. Let $y >_\pi A$ be any blue vertex. If $y \in B$, then we have $d(x,y) \leq 2D$, since $A \cup B = V(H_1)$ and $H_1$ was already checked. Otherwise, $y >_\pi C$. Then by assumption, we have checked that $d(x, b_1) \leq D$, $d(b_1, c_0) \leq M$, and $d(c_0, y) \leq D$, so we have verified that $d(x, y) \leq 2D+M$.

\end{enumerate}

By a symmetric argument, if $V_1$ is the set of vertices to the right of $B$, for any vertex $y \in V_1$ and any vertex $x \in V \setminus V_1$ having a different color than $y$, we have verified that $d(x, y) \leq 2D+M$.

Let $G_0 = G[V_0]$ and let $G_1 = G[V_1]$. We have verified that if $v, w$ are two differently colored vertices in $V$ which do not lie in the same $G_i$, then $\dm(v,w) \leq 2D+M$. We can then run the algorithm recursively on $G_0$ and $G_1$. We now split our analysis into the sparse case in which $m^{4/3}n^{1/3} < m^{1/2}n^{3/2}$ and the dense case in which the reverse inequality holds; in each case, the recursion adds a factor of $\log n$ to the runtime.

\textbf{Sparse case:}  We note $V_0 <_\pi B$ and in particular $V_0 <_\pi \{v_j\}$, so by our choice of $v_j$ there are at most $m/2$ edges in $G_0$. Likewise, $V_1 >_\pi \{v_j\}$, so there are at most $m/2$ edges in $G_1$. Then we have:
$$|E(G_0)|^{4/3}|V(G_0)|^{1/3}+ |E(G_1)|^{4/3}|V(G_1)|^{1/3} \leq 2(m/2)^{4/3}n^{1/3} \leq m^{4/3}n^{1/3}$$
Hence, the recursive runtime is $\tilde{O}(m^{4/3}n^{2/3})$.

\textbf{Dense case:} In this case, $V_0 <_\pi \{v_j\}$ and $V_1 >_\pi \{v_j\}$ where $j = \lceil \frac{n}{2} \rceil$, so there are at most $n/2$ edges in each of $G_0$ and $G_1$. We have:
$$|E(G_0)|^{1/2}|V(G_0)|^{3/2}+ |E(G_1)|^{1/2}|V(G_1)|^{3/2} \leq 2m^{1/2}(n/2)^{3/2} \leq m^{1/2}n^{3/2}$$
So the recursive runtime is $\tilde{O}(m^{1/2}n^{3/2})$.
\end{proof}

This yields the final min-diameter approximation algorithm by a standard binary search.

\begin{proof}[Proof of Theorem \ref{thm:full-bix-dag-algo}]
    Using Theorem \ref{bichdiam}, binary search over $D \in [0, M_1n]$ where $M_1$ is the maximum edge weight. This multiplies the runtime by $O(\log(M_1n))$. Since $M_1$ is polynomial in $n$, the overall runtime remains $\tilde{O}(\min(m^{4/3}n^{1/3}, m^{1/2}n^{3/2}))$. 
\end{proof}

\section{Linear time detection of infinite bichromatic min-diameter}\label{app:linear-finite}

 In this section, we give a linear-time algorithm which determines whether the bichromatic min-diameter of a (possibly weighted) directed graph is finite.
 
 We first prove the claim in the case of DAGs.

\begin{proposition}\label{bichfiniteDAG}
There is an $O(m)$ time algorithm which can determine whether an input DAG $G$ has finite bichromatic min-diameter.
\end{proposition}

\begin{proof}
Note that proving the claim for unweighted graphs is sufficient, since a weighted graph has finite bichromatic min-diameter only if the underlying unweighted graph also does. 

We topologically sort the graph and let $v_1, \dots v_n$ be its edges in sorted order $\pi$. We then partition the vertices into disjoint sets $A^1 <_\pi A^2 <_\pi \dots <_\pi A^r$ each of which is a maximal monochromatic closed set of vertices. The sets will necessarily alternate colors, so that $A^i$ is red for even indices $i$ and blue for odd $i$, or vice versa. Let $y^{k-1}$ be the left-most vertex in $A^k$, and let $x^k$ be the right-most vertex in $A^k$, for $k \in [r]$. Thus $x^k$ and $y^k$ are always consecutive vertices of different colors. Let $G^k$ be the subgraph induced by $\{x^{k-1}\} \cup A^k \cup A^{k+1} \cup \{y^{k+1}\}$. See Figure \ref{fig:bich-alg-inf}.

\begin{figure}[h]
    \centering
\includegraphics[scale=0.4]{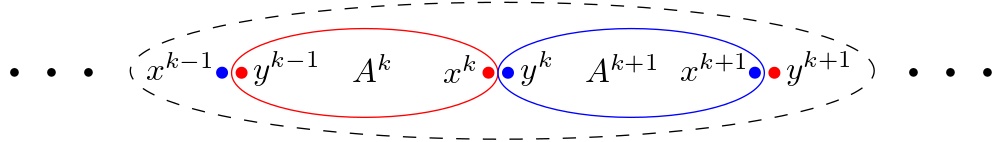}
\caption{The graph $G$, with the $k$th `sliding window' $G^k$ circled with a dashed line.}
\label{fig:bich-alg-inf}
\end{figure}
\medskip

Our algorithm will take a `sliding window' approach. We will check finiteness of red-blue min-distances between within the sliding window, and then show that if the check holds, all red-blue min-distances between vertices that have passed under the sliding window are finite.

Our inductive assumption will be that any min-distance between a red-blue pair of vertices in the set $L^k = A^1 \cup \cdots \cup A^k \cup \{y_k\}$ is finite. In the base case, $L^0 = \{y_1\}$, so the inductive assumption holds.

We will assume without loss of generality that in step $k$, $x^{k-1}$ and $A^{k+1}$ are blue, and $A^k$ is red; then for $i < k$, any red $v \in A^i$ can reach $x^{k-1}$ and any blue $v \in A^i$ can reach $y^{k-1}$. 

To simplify notation, we will write $B = A^{k+1}$.

Let $B_0 = \{b \in B \ | \ N^{in}(b) \cap B = \varnothing\}$ be the set of vertices in $B$ having no in-neighbors in $B$. Since $G$ is acyclic, for all $b \in B$ there is some $b_0 \in B_0$ such that $d(b_0, b)$ is finite.  We can compute $B_0$ in time $O(|E(B)| + |V(B)|)$, which is $O(|E(G^k)| + |V(G^k)|)$.


Now, we compute a set $A' \subseteq A^k$ as follows: for each $a \in A^k$, we check whether $a$ is adjacent to all of $B_0$. If so, we add $a$ to $A'$. We remark that any $a' \in A'$ can reach all vertices in $B$.

Next, we construct an auxiliary vertex $z^k$ with edges $(a', z^k)$ for all $a' \in A^k$, and we run BFS into $z^k$ within the graph $G[A^k]$. Suppose this BFS tree $BFS^{in}(z^k)$ hits some vertex $a \in A^k$. Then there is a path $a \to z^k$, and therefore since $N^{in}(z^k) = A'$, there is a path $a \to a'$ for some $a' \in A'$. So in particular, if $a \in BFS^{in}(z^k)$, then all vertices in $B$ are reachable from $a$. Hence, if $BFS^{in}(z^k) = A^k$, all vertices in $B$ are reachable from all vertices in $A^k$.

Otherwise, let $a \in A^k$ be the right-most vertex in $A^k$ not in $BFS^{in}(z^k)$. Then in particular, $a$ has no out-neighbors in $A^k$.  Furthermore, since $a \not\in A'$, there is some $b_0 \in B_0$ such that $a$ is not adjacent to $b_0$. Since $b_0$ has no in-neighbors in $B$,  there can be no path from $a$ to $b_0$, and the graph therefore has infinite bichromatic min-diameter. In this case we halt the algorithm and report this result.

Determining if $a \in A'$ takes time $O(\deg(a))$ for each $a$, via checking the adjacency list to see if it includes all of $B_0$. So altogether the runtime needed to check all vertices in $A^k$ is at most $O(\Sigma_{a \in A^k} \deg(a))$. This runtime is at most $O(|E(G^k)|)$. Likewise, running a BFS from the auxiliary vertex $z^k$ within $G[A^k]$ takes time at most $O(|E(G^k)| + |V(G^k)|)$.

We have now verified that all vertices in $B = A^{k+1}$ are reachable from all vertices in $A^k$. We will furthermore run BFS from $x^{k-1}$ and $y^{k+1}$ within the subgraph $G^k$. This takes time at most $O(|E(G^k)| + |V(G^k)|)$. If any red vertex in $G^k$ is not reachable from $x^{k-1}$, or if $y^{k+1}$ is not reachable from some blue vertex in $G^k$, we halt and report that the bichromatic min-diameter of $G$ is infinite. 

Otherwise, we have that all red-blue min-distances between vertices in $\{x^{k-1}\} \cup A^k \cup A^{k+1} \cup \{y^{k+1}\}$ are finite. We furthermore have inductively that all red-blue min-distances within $L^k = A^1 \cup \cdots \cup A^{k} \cup \{y_{k}\}$ are finite. 
Let $i < k$ and let $v \in A^i$ be any blue vertex. Then there is a path $v \to x^k \to y^k \to y^{k+1}$. Likewise, let $w \in A^i$ be any red vertex. Then for any $b \in B = A^{k+1}$, there is a path $w \to x^{k-1} \to y^{k-1} \to b$. Thus, all red-blue min-distances between vertices in $A^i$ for $i < k$ and vertices in $A^{k+1} \cup \{y^{k+1}\}$ are finite.

Thus all red-blue min-distances within $L^{k+1} = A^1 \cup \cdots \cup A^{k+1} \cup \{y^{k+1}\}$ are finite, completing the inductive step $k$.

Since every edge appears in at most three subgraphs $G^k$, the total runtime is at most:

$$O\left( \sum_{k} |E(G^k)| + |V(G^k)| \right) = O(m).$$

\end{proof}

We can now prove the theorem for general directed graphs.

\begin{reminder}{Theorem \ref{thm:bichr-finite}}
There is an $O(m)$ time algorithm which checks, for any weighted directed graph $G$, whether the bichromatic min-diameter is finite.
\end{reminder}

\begin{proof}

It suffices to prove the claim for unweighted $G$.

Using Kosaraju's algorithm, one can in $O(m)$ time compute the strongly connected components $C_1, C_2,$ $\dots, C_t$ of $G$. We then construct a DAG $G'$ whose vertices, singly or in pairs, correspond to the strongly connected components of $G$, as follows: For each SCC $C_i$, construct a vertex $x_i$. For $i \neq j$, add an edge $x_i \to x_j$ if there is at least one edge $C_i \to C_j$. If $C_i$ consists entirely of red (blue) vertices, color $x_i$ red (blue). Otherwise, $C_i$ contains both red and blue vertices; in this case, we split $x_i$ into two vertices, $x_i^R$ and $x_i^B$, where $x_i^R$ retains all edges incoming into $x_i$, $x_i^B$ retains all edges outgoing from $x_i$, and there is an edge $x_i^R \to x_i^B$. Informally, there are edges $N^{in}(C_i) \to x_i^R \to x_i^B \to N^{out}(C_i)$.

Using Proposition \ref{bichfiniteDAG}, we can in time $O(m)$ check whether $G'$ has finite bichromatic min-diameter. If not, then there are some SCCs $C_i, C_j$ of $G$, one of which contains a red vertex and the other containing a blue vertex, such that there is no path from any vertex in $C_i$ to any vertex in $C_j$ or vice versa. So in this case we report that $G$ likewise has infinite bichromatic min-diameter.

Otherwise, we have that for all SCCs $C_i, C_j$ of $G$, one of which contains a red vertex $a$ and the other of which contains a blue vertex, $b$, that there is some path from $C_i$ to $C_j$ or vice versa. Since $C_i$ and $C_j$ are strongly connected, this implies the existence of a path from $a$ to $b$ or vice versa. Moreover, for any red-blue pair $a, b$ within the same SCC, there is a path from $a$ to $b$, by the definition of strongly connectedness. Thus, in this case we report that $G$ has finite bichromatic min-diameter.
\end{proof}

\section*{Acknowledgements}

The authors are grateful to our reviewers for many helpful suggestions.


\bibliography{references}

\appendix

\section{$\frac{3}{2}$-approximation for min-diameter in DAGs}\label{app:3/2-mindiam}

In this appendix we describe how to modify the algorithm of Theorem \ref{thm:3/2-unweighted-dag} in order to achieve an exact $\frac{3}{2}$-approximation. We additionally give two variations of the algorithm. The first is a combinatorial version which avoids the use of fast matrix multiplication and has a slightly slower runtime. The second is a version which partitions the vertex set into subsets of size $O(nk^2/m)$ rather than $O(k^2)$ as in the main algorithm, which runs slightly faster when $m \leq n^{1.283}$.

\begin{reminder}{Theorem \ref{thm:exact}}
There is an $\tilde{O}(m^{\frac{8\beta+4-4\alpha\beta}{5\beta+3-2\alpha\beta}}n^{\frac{\beta+1}{5\beta+3-2\alpha\beta}})$-time algorithm achieving a $\frac{3}{2}$-approximation for min-diameter in unweighted DAGs.
\end{reminder}

Given a DAG $G$ with $n$ vertices and $m$ edges, the first step in the algorithm involves constructing an auxiliary graph $G'$ by subdividing every edge of $G$. Then all vertices in $V' = V(G')$ are of the form $v \in V = V(G)$ or $v = v_{x,y}$ for some $x, y \in V$, where the latter is the vertex in the subdivided edge $(x,y)$, with $x \to v_{x,y} \to y$. Then $G'$ has $m+n = O(m)$ vertices and $2m = O(m)$ edges. We scale down all edge weights by a factor of 2, so that all distances between vertices in $V \subset V'$ remain the same.
A topological ordering $\pi_0$ of $G$ can be extended to a topological ordering $\pi$ of $G$. We note that for a vertex $v$, $\epsilon(v)$ refers to the eccentricity of $v$ within the original graph $G$, and $N^{out}_{d,S; G'}(v)$ (or any sets or quantities marked with a $'$ symbol) are considered within the auxiliary graph $G'$. Since $d(u,v) = d'(u,v)$ whenever $u,v \in V$, we use the notation $d(\cdot, \cdot)$ to refer to distance in either $G$ or $G'$.

From here, the algorithm and its proof follow that of Theorem \ref{thm:3/2-unweighted-dag}, with two key differences. One is that we eliminate the ceiling function everywhere it appears. The other is we construct a $(k,D/2)$-neighborhood cover of \textit{edges} in $G$ instead of vertices, by way of covering the auxiliary vertices $v_{x,y}$ inside of subdivided edges in the graph $G'$.

This set $S'$ then functions in just the same way as a usual neighborhood cover for purposes of set intersections throughout the proof; the ceiling function is no longer necessary because every path of length $\leq D$ between vertices $a \in A, b\in B$ now has an \textit{exact} midpoint which is at most $D/2$ from each end.

The set $S'$ functions slightly differently for cases where we use vertices $s \in S$ as jumping-off points, because we cannot assume that an auxiliary vertex $v_{x,y}$ has min-eccentricity at most $D$ if the original graph $G$ has min-diameter less than $D$. So instead, for each $s = v_{x,y} \in S$ corresponding to an edge, we BFS to and from the endpoints $x, y$ of the edge, and we report than the min-diameter is greater than $D$ if $\epsilon(x)$ or $\epsilon(y)$ is greater than $D$. Otherwise, we can use $x$ or $y$ as a jumping-off point anywhere where we would have used $v_{x,y}$: If $d(a,v_{x,y}) \leq D/2$ (without loss of generality), then since the $a \to v_{x,y}$ path must pass through $x$, we have $d(a,x) \leq D/2$ as well. A symmetric argument holds for paths outgoing from $v_{x,y}$; these must all pass through $y$. Therefore, at any point in the argument where we would originally have used the fact that $\epsilon(v_{x,y}) \leq D$, we can instead rely on the fact that $\epsilon(x) \leq D$ and $\epsilon(y) \leq D$.

\begin{algorithm}[h]
    \KwIn{DAG $G = (V,E)$ and corresponding auxiliary DAG $G' = (V', E')$, topological ordering $\pi'$ of $G'$, subsets $A <_\pi B \subseteq V \subset V'$ with $|A|=|B|=p=O(nk^2/m)$, diameter guess $D$, $(k, D/2)$-neighborhood cover $S' \subseteq V'$ such that $\epsilon(s) \leq D$ for all $s \in S' \cap V$ and $\epsilon(x), \epsilon(y) \leq D$ for all $v_{x,y} \in S' \setminus V$, and $N^{out}_{D/2,S; G'}(a)$, $N^{in}_{ D/2, S; G'}(b)$ for $a \in A, b \in B$.}
    \KwOut{One of the following. Each output verifies a corresponding property of $G$.
    \begin{tabular}{ccc}
		\midrule
		\Pass  & $\Rightarrow$ & $\dm(a,b) \le  \frac 32 D  $ for all in $a \in A$ and $b \in B$  \\
		\Fail & $\Rightarrow$ & $\MinDiam(G) > D$\\
	\end{tabular}}
    Compute $M_A$, the matrix with columns given by indicator vectors of $N^{out}_{D/2, S';G'}(a)$ for $a \in A$\;
    Compute $M_B$, the matrix with columns given by indicator vectors $N^{in}_{ D/2 , S';G'}(b)$ for $b \in B$\;
    Compute $M=M_A^TM_B$\;
    \ForEach{$a \in A, b \in B$}{ 
        \uIf{ $M_{ab} \ge 1$\label{line:cond-0-exact}}
        {\Continue\;}
        \uElseIf{$a <_\pi N^{in}_{ D/2, S';G'}(b)$ \andt
        $N^{in}_{ D/2 , S';G'}(b) \cap S \neq \emptyset$\label{line:cond-1-exact}}
        {\Continue\; } 
        \uElseIf{$b >_\pi N^{out}_{D/2;G'}(a)$ \andt $N^{out}_{D/2;G'}(a) \cap S \neq \emptyset$\label{line:cond-2-exact}}
        {\Continue\;}
        \Else{\Fail\;}
    }
    \Pass\;
\caption{All-pairs exact min-distance tester}\label{alg:all-pairs-exact}
\end{algorithm}

\begin{lemma}
\cref{alg:all-pairs-exact} produces the correct output in runtime $O(k^{4 + \frac{2\beta - 2\alpha\beta}{\beta + 1}})$.
\end{lemma}
\begin{proof}
The proof follows that of Lemma \ref{lemma:all-pairs}. If the algorithm fails, it fails inside the \textbf{foreach} loop at some pair $a,b$. Suppose for the sake of contradiction that $\dm(a,b) \leq D$. Then because we have subdivided the edges, there is some $x$ such that $d(a,x), d(x,b) \le  D/2$. We then arrive at a contradiction by the same arguments given in the proof of Lemma \ref{lemma:all-pairs}.

Now, assume the algorithm passes. Then for each pair $a \in A, b \in B$, one of the conditions in lines \ref{line:cond-0-exact}, \ref{line:cond-1-exact}, or \ref{line:cond-2-exact} must hold. If $M_{ab} \geq 1$, then there is some $x \in N^{out}_{D/2, S';G'}(a) \cap N^{in}_{D/2, S';G'}(b)$, giving $d(a,b) \le d(a,x) + d(x,b) \le D/2 + D/2 \leq D$. Otherwise if $a <_\pi N^{in}_{ D/2, S';G'}(b)$ and $N^{in}_{ D/2, S';G'}(b) \cap S' \neq \emptyset$, then there is some $s \in N^{in}_{ D/2, S';G'}(b) \cap S'$ with $s \geq_\pi a$. If $s \in V$, $\epsilon(s) \leq D$, and $d(a,b) \le d(a,s) + d(s,b) \le D +  D/2  \le  3D/2$.  If $s = v_{xy} \in S' \setminus V$, then $d(y,b) < d(s, b) \leq D/2$, and then since $\epsilon(y) \leq D$, we have $d(a,b) \leq d(a,y) + d(y,b) \leq D + D/2 = 3D/2$. A symmetric argument holds in the case that $b >_\pi N^{out}_{D/2, S';G'}(a)$ and $N^{out}_{D/2, S';G'}(a) \cap S' \neq \emptyset$. We conclude that if the algorithm passes, then every pair $a \in A, b \in B$ satisfies $d(a,b) \le  3D/2 $.

Runtime analysis is as in the proof of Lemma \ref{lemma:all-pairs}.
\end{proof}

\begin{lemma}\label{ugly}
For $k \geq (m/n)^{\frac{\omega-1}{2(\omega-2)}}$, Algorithm \ref{alg:all-pairs-exact} can be modified to produce the specified output in time $O((n/m)^{\frac{2\beta + 2 - \alpha\beta}{\beta+1}}k^{{\frac{6\beta+4-2\alpha\beta}{\beta+1}}})$. 
\end{lemma}

\begin{proof}

It is not necessary to fix the size $p$ of the sets $|A_i|, |B_i|$ to be $O(k^2)$. The runtime is in general $O((pk)^{\frac{2\beta}{\beta+1}}p^\frac{2-\alpha\beta}{\beta+1} + p^2) = O(p^{\frac{2\beta + 2 - \alpha\beta}{\beta+1}}k^{\frac{2\beta}{\beta+1}} + p^2)$.

We will now set $p = nk^2/m$, which we can do since $k \geq \sqrt{m/n}$. In order to apply Lemma \ref{rectangular}, as is needed for the fast matrix multiplication, we require that $p^{1+\frac{\alpha}{2}} \leq pk \leq p^{\frac{\omega+1}{2}}$, for current bounds on $\alpha, \omega$. The former inequality holds as $p^\frac{\alpha}{2} = (nk^2/m)^{\alpha/2} < k$ since $\alpha < 1$. The latter inequality holds for $k \leq (nk^2/m)^\frac{\omega-1}{2}$, which is true for $k \geq (m/n)^{\frac{\omega-1}{2(\omega-2)}}$.

This runtime then becomes: $$O((n/m)^{\frac{2\beta + 2 - \alpha\beta}{\beta+1}}k^{{\frac{6\beta+4-2\alpha\beta}{\beta+1}}} + n^2m^{-2}k^4) = O((n/m)^{\frac{2\beta + 2 - \alpha\beta}{\beta+1}}k^{{\frac{6\beta+4-2\alpha\beta}{\beta+1}}}).$$
\end{proof}

\begin{remark}
Algorithm \ref{alg:all-pairs-exact} can also be modified into a combinatorial algorithm, which produces the same output without use of fast matrix multiplication. We instead check set intersections in $\tilde{O}(p^2k)$ time by brute force; the runtime is then $\tilde{O}(p^2k)$. If $k \geq \sqrt{m/n}$, we can set $p = nk^2/m$ to obtain a runtime of $\tilde{O}(n^2k^5/m^2)$. 
\end{remark}

\begin{algorithm}[h]
    \KwIn{DAG $G = (V,E)$ and corresponding auxiliary DAG $G' = (V', E')$, topological ordering $\pi'$ of $G'$, subsets $A <_\pi B \subseteq V \subset V'$ with $|A|=|B|=p=O(k^2)$, diameter guess $D$, $(k, D/2)$-neighborhood cover $S' \subseteq V'$ such that $\epsilon(s) \leq D$ for all $s \in S' \cap V$ and $\epsilon(x), \epsilon(y) \leq D$ for all $v_{x,y} \in S' \setminus V$, and neighborhoods $N^{out}_{D/2,S'; G'}(a)$ for all $a \in A$.}
    \KwOut{One of the following. Each output verifies a corresponding property of $G$.
    \begin{tabular}{ccc}
		\midrule
		\Pass and return $A$ & $\Rightarrow$ & $d(a,v)\le\frac 32D$ for all $a \in A$ and $v >_\pi B$  \\
		\Pass and return $B$ & $\Rightarrow$ & $d(v,b)\le  \frac 32D $ for all $b \in B$ and $v <_\pi A$ \\
		\Fail & $\Rightarrow$ & $\MinDiam(G) > D$\\
	\end{tabular}}
    \ForEach{$a \in A$}{
        \If{ $N^{out}_{D/2, S'; G'}(a) \cap S' = \emptyset$ or $N^{out}_{D/2, S';G'}(a) \cap S' >_\pi B$\label{line:big-out-nbd-exact}}{
            \ForEach{ $v \in \{a\} \cup N^{out}_{D/2, S';G'}(a)$}{
                \uIf{ $v \in V$}
                {BFS to and from $v$\; 
                    \If{$\epsilon(v) > D$\label{line:fail-cond-2-exact}}
                {\Fail\;}
                }
                \Else{ Must have $v = v_{xy} \in V' \setminus V$\;
                    BFS to and from $y$\; 
                    \If{$\epsilon(y) > D$\label{line:fail-cond-3-exact}}
                {\Fail\;}
                }
            }
            \Pass and \Return $B$\;
            
            }
        
    }
    \Pass and \Return $A$\;
\caption{Directional exact min-distance tester}\label{alg:directional-exact}\end{algorithm}
\begin{lemma}\label{lemma:directional-exact}
Algorithm \ref{alg:directional-exact} produces the correct output in runtime $\tilde{O}(mk)$.
\end{lemma}
\begin{proof}
We first show correctness. If the algorithm fails, we must have found a vertex $v$ with $\epsilon(v) > D$, and thus $\MinDiam(G) > D$.

If the algorithm returns $A$, then for each $A \in a$, there is some $s \in S' \cap N^{out}_{D/2, S;G'}(a)$ with $s$ appearing in $B$ or to its left. The argument proceeds as above: if $s \in V$, then $\epsilon(s) \leq D$, and so for each $b >_\pi B$ we have $d(a,b) \le d(a,s) + d(s,b) \le D +  D/2  \le  3D/2$.  If $s = v_{xy} \in S' \setminus V$, then $d(y,b) < d(s, b) \leq D/2$, and then since $\epsilon(y) \leq D$, we have $d(a,b) \leq d(a,y) + d(y,b) \leq D + D/2 = 3D/2$.

Otherwise, the algorithm returns $B$. Then the condition in line \ref{line:big-out-nbd-exact} must hold for some $a \in A$. Let $c <_\pi A$ and $b \in B$. Since we do not fail in line \ref{line:fail-cond-2-exact}, we must have $d(a,b) \le D$. Let $v$ be a midpoint of the shortest path from $a$ to $b$, so $d(a,v), d(v,b) \leq D/2$. Since $N^{out}_{D/2,S;G'}(a)$ is either not cut off by hitting $S$ or is cut off after $B$, we have $v \in N^{out}_{D/2}(a)$. If $v \in V$, since we do not fail in line \ref{line:fail-cond-2-exact}, we must have $d(c,v) \le D$. Concluding, $d(c,b) \le d(c,v)+d(v,b) \le D + \lceil D/2 \rceil \le \lceil 3D/2 \rceil$. Alternatively, iv $v = v_{xy} \in V' \setminus V$, since we do not fail in line \ref{line:fail-cond-3-exact}, we have $d(c, y) \le D$. Then $d(c, b) \le d(c, y) + d(y, b) \le D + D/2 \le 3D/2$.
As $c,b$ were arbitrary, this completes the case.

We conclude with runtime analysis. The outer \textbf{foreach} loop repeats until we have covered every vertex in $A$ or the condition in line \ref{line:big-out-nbd-exact} is satisfied. Checking this condition takes time at most $O(k)$ for a total of $\tilde O(|A|k) = O(mk)$. If the condition is satisfied, we perform $1+|N^{out}_{D/2,S}(a)|$ calls to BFS, for a total of $\tilde O(mk)$. The algorithm concludes before returning to the outer loop, so we may add these contributions for a total time of $\tilde O(mk).$
\end{proof}
\begin{algorithm}[h]
    \KwIn{DAG $G = (V,E)$, diameter guess $D$. Parameters $k$ and $\theta = n/2k^2$}
    \KwOut{One of the following. Each output verifies a corresponding property of $G$.
    \begin{tabular}{ccc}
		\midrule
		\Pass  & $\Rightarrow$ & $\MinDiam(G) \le \frac 32 D $  \\
		\Fail & $\Rightarrow$ & $\MinDiam(G) > D$\\
	\end{tabular}}
    Construct auxiliary graph $G'$ by subdividing edges of $G$\;
    Topologically sort $G$ and extend this ordering to $G'$\;
    Using Lemma \ref{nbrhds} compute a $(k,D/2)$-neighborhood cover $S' \subseteq V'$. Run BFS within $G$ to and from each vertex in the set $T = (S' \cap V) \cup N^{in;G'}(S' \setminus V) \cup N^{out;G'}(S' \setminus V)$\;
    \If{ $\exists t \in T$ such that $\epsilon(t) > D$}{\Fail}
    Partition $V$ into consecutive closed intervals, $V_0, \dots V_{2\theta}$, with $|V_i| = n/(2\theta) = k^2$ for each $i$\;
    \tcc{Let $L,R$ be the left and right halves of $V$, so $L = V_0 \cup \cdots \cup V_{\theta}$ and $R = V_{\theta+1} \cup \cdots \cup V_{2\theta}$.
    We first check all min-distances between L and R.}
    Initialize $i = \theta$ and $j = \theta + 1$\;
    \While{$i \ge 1 \textbf{ and } j \le 2\theta$}{
        Run \cref{alg:all-pairs-exact} (all-pairs) on the pair $(V_i, V_j)$. If this fails, then \Fail\;
        Run \cref{alg:directional-exact} (directional tester) on the pair $(V_i, V_j)$. If this fails, then \Fail\;
        \uElseIf{\cref{alg:directional-exact} passes and returns $V_i$}{$i = i - 1$\;
        }
        \Else{
        \tcc{Otherwise, it must be that \cref{alg:directional-exact} passes and returns $V_j$}
        $j = j + 1$\;\label{line:end-while-loop-exact}
        }
    }
    \tcc{If this line is reached then all distances from $L$ to $R$ are at most $ \frac 32 D$.}
    Recursively call this algorithm on $G[L]$ and $G[R]$. If either fails, then \Fail. Else \Pass\;
        
\caption{Full graph min-distance tester}\label{alg:Full-graph-min-distance-tester-exact}
\end{algorithm}

\begin{lemma}\label{lemma:full-graph-min-distance-tester-exact}
Algorithm \ref{alg:Full-graph-min-distance-tester-exact} produces the correct output in runtime $\tilde{O}(m^{\frac{8\beta+4-4\alpha\beta}{5\beta+3-2\alpha\beta}}n^{\frac{\beta+1}{5\beta+3-2\alpha\beta}})$. 
\end{lemma}
\begin{proof}
We note that $T \subseteq V$, since all neighbors of vertices in $V' \setminus V$ lie in $V$.

The algorithm fails only when some $t \in T$ has $\epsilon(t) > D$, when \cref{alg:all-pairs-exact} fails, or \cref{alg:directional-exact} fails, all of which imply $\MinDiam(G) > D$. In the event of a pass, $\MinDiam(G) \le 3D/2 $; the proof of this claim is the same as in Lemma \ref{alg:Full-graph-min-distance-tester}.

$G'$ can be constructed in $O(m)$ time, and the topological sorting takes time $O(m)$. Lemma \ref{nbrhds} constructs the set $S'$ in time $O(nk^2)$, as it suffices to cover neighborhoods of $V$. Running BFS to and from each vertex in $T$, which has size $|T| \leq 2|S'| = O(m/k)$, takes time $\tilde{O}(m^2/k)$. We run Algorithm \ref{alg:all-pairs-exact} and Algorithm \ref{alg:directional-exact} each up to $2\theta = n/2k^2$ times. Since Algorithm \ref{alg:all-pairs-exact} takes time $O(k^{4 + \frac{2\beta - 2\alpha\beta}{\beta + 1}})$ and Algorithm \ref{alg:directional-exact} takes time $\tilde{O}(mk)$, the total runtime of a recursive step is therefore:
$$\tilde{O}(m^2/k + nk^{2 + \frac{2\beta - 2\alpha\beta}{\beta + 1}})$$

Setting $k = (m^2/n)^{\frac{\beta + 1}{5\beta + 3 - 2\alpha\beta}}$, we obtain $\tilde{O}(m^{\frac{8\beta+4-4\alpha\beta}{5\beta+3-2\alpha\beta}}n^{\frac{\beta+1}{5\beta+3-2\alpha\beta}})$ for each recursive step. The recursion then adds a logarithmic factor.
\end{proof}

\begin{lemma}\label{ugly-2}
Algorithm \ref{alg:Full-graph-min-distance-tester-exact} can be modified to produce the specified output in runtime $\tilde{O}(m^{\frac{7\beta + 3 -3\alpha\beta}{5\beta+3-2\alpha\beta}}n^{\frac{2\beta+2-\alpha\beta}{5\beta+3-2\alpha\beta}})$ if $m \leq n^{1.283}$.
\end{lemma}

\begin{proof}
The proof is identical to that of Lemma \ref{lemma:full-graph-min-distance-tester-exact}, except that we set $p = nk^2/m$ and $\theta = n/2p = m/2k^2$, and we use the modified version of Algorithm \ref{alg:all-pairs-exact} given in Lemma \ref{ugly}. We assume for now that $k$ satisfies the conditions of Lemma \ref{ugly}.

The runtime analysis then changes as follows: Running BFS to and from vertices in $T$ still takes time $\tilde{O}(m^2/k)$. Running Algorithm \ref{alg:directional-exact} $n/p = m/2k^2$ times likewise takes time $\tilde{O}(m^2/k)$. 
Running the modified version of Algorithm \ref{alg:all-pairs-exact} $m/2k^2$ times takes time $O\left(n^{\frac{2\beta + 2 - \alpha\beta}{\beta+1}}m^{\frac{-\beta-1+\alpha\beta}{\beta+1}}k^{{\frac{4\beta+2-2\alpha\beta}{\beta+1}}}\right)$.
Setting $k = m^{\frac{3\beta+3-\alpha\beta}{5\beta+3-2\alpha\beta}}n^{\frac{-\beta-2+\alpha\beta}{5\beta+3-2\alpha\beta}}$ gives a runtime of $\tilde{O}(m^{\frac{7\beta + 3 -3\alpha\beta}{5\beta+3-2\alpha\beta}}n^{\frac{2\beta+2-\alpha\beta}{5\beta+3-2\alpha\beta}})$.  


It remains to check that $k$ satisfies the conditions of Lemma \ref{ugly}. Substituting $\alpha > 0.31389$ \cite{lu18}, $\omega < 2.37286$ \cite{matrixmult2020}, and $\beta \simeq 0.5435$, we can verify that $k = m^{\frac{3\beta+3-\alpha\beta}{5\beta+3-2\alpha\beta}}n^{\frac{-\beta-2+\alpha\beta}{5\beta+3-2\alpha\beta}}$ is greater than $(m/n)^\frac{\omega-1}{2(\omega -2)}$ when $m \leq n^{1.283}$.

\end{proof}
\begin{proof}[Proof of Theorem \ref{thm:exact}]

Binary searching over $D \in [n]$ using Algorithm \ref{alg:Full-graph-min-distance-tester-exact} gives the desired $\frac{3}{2}$-approximation algorithm.

\end{proof}


\begin{remark}
When $m \leq n^{1.283}$, Lemma \ref{ugly-2} gives a $\frac{3}{2}$-approximation algorithm for min-diameter in DAGs running in time $O(m^{1.171}n^{0.543})$.
\end{remark}

\begin{remark}
As mentioned above, the brute-force set-intersection version of \cref{alg:all-pairs-exact} runs in time $\tilde{O}(p^2k) = \tilde{O}(n^2k^5/m^2)$. Taking $\theta = m/2k^2$ and $k = m^{3/4}/n^{1/2}$ yields a combinatorial version of \cref{alg:Full-graph-min-distance-tester-exact} which runs in time  $\tilde{O}(m^{5/4}n^{1/2})$.
\end{remark}

\end{document}